\documentclass[conference]{IEEEtran}
\usepackage{etex}
\usepackage{dsfont}
\usepackage{cite}
\usepackage{times}
\usepackage{helvet}
\usepackage{courier}
\usepackage[utf8]{inputenc} 
\usepackage{fixltx2e} 
\usepackage{breakurl} 
\usepackage{mathtools}
\usepackage{complexity}
\usepackage{balance}
\usepackage{mathdots}
\usepackage{bbold}
\RequirePackage{latexsym}
\RequirePackage{mathtools}
\RequirePackage{amsfonts}
\usepackage{amssymb}
\RequirePackage{amsthm}
\RequirePackage{array}
\RequirePackage{xspace}
\usepackage{url}
\usepackage{graphicx}

\newcommand{\ie}{i.e., }

\newcommand{\etal}{et al.\xspace}
\newcommand{\wrt}{w.r.t.\ }

\let\oldnl\nl
\newcommand{\nonl}{\renewcommand{\nl}{\let\nl\oldnl}}

\usepackage[textsize=small]{todonotes}

\newcommand{\RR}{\mathbf{R}\xspace}

\newcommand{\N}{\mathbf{N}\xspace}
\newcommand{\Z}{\mathbf{Z}\xspace}

\newcommand{\STN}{STN\xspace}
\newcommand{\MPG}{MPG\xspace}

\newcommand{\HTN}{HyTN\xspace}
\newcommand{\DTP}{DTP\xspace}
\newcommand{\HTNC}{HyTN-Consistency\xspace}
\newcommand{\CSTP}{CSTP\xspace}
\newcommand{\CSTN}{CSTN\xspace}
\newcommand{\eDCC}{CSTN-$\epsilon$-DC\xspace}
\newcommand{\DCC}{CSTN-DC\xspace}

\newcommand{\scHst}{\emph{scHst}\xspace}
\newcommand{\Con}{\emph{Con}\xspace}
\newcommand{\Sub}{\emph{Sub}\xspace}

\def\Ord{{\cal O}}
\def\C{{\cal C}}

\def\H{{\cal H}}

\def\A{{\cal A}}

\def\S{{\cal S}}

\newcommand{\figref}[1]{Fig.~\ref{#1}}

\newtheorem{Thm}{Theorem}

\newtheorem{Lem}[]{Lemma}

\theoremstyle{plain}
\newtheorem{Def}{Definition}
\newtheorem{Rem}{Remark}
\newtheorem{Ex}{Example}

\usepackage[ruled,vlined,linesnumbered]{algorithm2e}
\usepackage{subfig}
\SetCommentSty{textsf}
\SetKwRepeat{DoWhile}{do}{while}
\SetAlFnt{\small} 
\makeatletter
\newcommand{\removelatexerror}{\let\@latex@error\@gobble}
\makeatother

\usepackage{tikz}
\usepackage{tikz-qtree}
\usetikzlibrary{calc,positioning,fit}
\usetikzlibrary{shapes,shapes.multipart,shapes.arrows}
\usetikzlibrary{decorations,decorations.pathmorphing,decorations.pathreplacing,decorations.markings,decorations.shapes}
\usetikzlibrary{arrows}
\usetikzlibrary{fit,backgrounds}
\usetikzlibrary{plotmarks}
\tikzstyle{node}=[circle,draw,inner sep=2pt,transform shape,minimum size=1.75em]
\tikzstyle{task}=[draw,rectangle,inner sep=1.5pt,transform shape
                         ,font=\sffamily\scriptsize,text width=27pt
                         ,text centered,minimum height=2em
                         ]
\tikzstyle{connector}=[draw,diamond,shape aspect=1,inner sep=1pt,transform shape
                 ,font=\sffamily\scriptsize,text width=15pt
                 ,text centered
                 ]
                                  
\newcommand{\AND}{\tikz[baseline,arrows=-] \draw[thick] (0,2ex) -- (4ex,2ex) (2ex,0) -- (2ex,4ex);}

\tikzstyle{StartCase}=[circle,draw,minimum size=.75cm,transform shape]
\tikzstyle{EndCase}=[circle,draw,ultra thick,minimum size=.75cm,transform shape]
\tikzstyle{smallLabel}=[font=\sffamily\scriptsize,inner sep=1pt,transform shape]
\tikzstyle{timeLabel}=[smallLabel,midway,transform shape]
\tikzstyle{minWidth}=[text width=.9cm]
\tikzstyle{info}=[rounded corners,fill=yellow,text width=1cm,text centered,inner sep=1pt]
\tikzstyle{infoLine}=[thin,decorate,decoration={snake,amplitude=.4mm, segment length=2mm, post length=1mm}]
\tikzstyle{edgeLabel}=[font=\tiny,sloped]
\tikzstyle{infoRow}=[rounded corners,fill=yellow,inner sep=1pt]
\tikzstyle{punto}=[circle,draw,fill=black,minimum size=2bp,inner sep=0pt,outer sep=0pt]
\tikzstyle{blackNode}=[fill=black!30]
\tikzstyle{crosses}=[decorate,decoration={name=crosses,segment length=2mm,post length=2mm}]
\tikzstyle{every picture}=[>=latex]
\tikzstyle{every label}=[inner sep=2pt]

\title{Dynamic Consistency of Conditional Simple Temporal 
Networks via Mean Payoff Games: \\ {\LARGE a Singly-Exponential Time DC-Checking}}
\author{\IEEEauthorblockN{Carlo Comin
}
\IEEEauthorblockA{Department of Mathematics \\ University of Trento, Trento, Italy\\ 
carlo.comin@unitn.it
}
\and
\IEEEauthorblockN{Romeo Rizzi}
\IEEEauthorblockA{Department of Computer Science \\ University of Verona, Verona, Italy\\ 
romeo.rizzi@univr.it}
}

\IEEEoverridecommandlockouts
\begin{document}
\maketitle

\begin{abstract}
Conditional Simple Temporal Network (\CSTN) is a constraint-based graph-formalism for conditional temporal planning. 
It offers a more flexible formalism than the equivalent CSTP model of Tsamardinos, Vidal and Pollack, 
from which it was derived mainly as a sound formalization.
Three notions of consistency arise for \CSTN{s} and CSTPs:
weak, strong, and dynamic. Dynamic consistency is
the most interesting notion, but it is also the most challenging and it was conjectured to be hard to assess.
Tsamardinos, Vidal and Pollack gave a doubly-exponential time algorithm
for deciding whether a \CSTN is dynamically-consistent and to produce, 
in the positive case, a dynamic execution strategy of exponential size. 
In the present work we offer a proof that deciding whether
a \CSTN is dynamically-consistent is \coNP-hard and provide
the first singly-exponential time algorithm for this problem,
also producing a dynamic execution strategy whenever the input \CSTN is dynamically-consistent. 
The algorithm is based on a novel connection with Mean Payoff Games, 
a family of two-player infinite games played on finite graphs,    
well known for having applications in model-checking and formal verification. 
The presentation of such connection is mediated by the Hyper Temporal Network model, 
a tractable generalization of Simple Temporal Networks 
whose consistency checking is equivalent to determining Mean Payoff Games. 
In order to analyze the algorithm we introduce a refined notion of dynamic-consistency, 
named $\epsilon$-dynamic-consistency, 
and present a sharp lower bounding analysis on the critical value of the reaction
time $\hat{\varepsilon}$ where the \CSTN transits from being, to not being, dynamically-consistent.
The proof technique introduced in this analysis of $\hat{\varepsilon}$
is applicable more generally when dealing with linear difference constraints which include strict inequalities.
\end{abstract}

\section{Introduction and Motivation}\label{sect:introduction}

In \emph{temporal planning} and \emph{temporal scheduling},
\emph{Simple Temporal Networks} (\STN{s}) \cite{DechterMP91} 
are directed weighted graphs, where nodes represent events to be scheduled in time  
and arcs represent temporal distance constraints between pairs of events. 
Recently, \STN{s} have been generalized into \emph{Hyper Temporal Networks} (\HTN{s}) \cite{CPR2014, CPR2015} by considering
weighted directed hypergraphs, where each hyperarc models a \emph{disjunctive} temporal constraint called \emph{hyper-constraint}. 
The computational equivalence between checking the consistency of {\HTN}s 
and determining winning regions in \emph{Mean Payoff Games}~({\MPG}s)~\cite{EhrenfeuchtMycielski:1979, ZwickPaterson:1996, brim2011faster} 
was pointed out as well in~\cite{CPR2014, CPR2015},   
where the approach was shown to be robust thanks to extensive experimental evaluations~\cite{CPR2014, CPR2015, BC12}.
Mean Payoff Games are a family of two-player 
infinite games played on finite graphs, 
well known for having theoretical interest in computational complexity, 
being it one of the few (natural) problems lying in  
$\NP\cap\coNP$, as well as various applications in model-checking 
and formal verification~\cite{Gradel2002}. 

The present work unveils that \HTN{s} and \MPG{s} are a natural 
underlying combinatorial model for checking the dynamic-consistency of conditional temporal problems.
We focus on \emph{Conditional Simple Temporal Problems} (\CSTP) \cite{TVP2003} 
and on their graph-based counterpart \emph{Conditional Simple Temporal Networks} (\CSTN) \cite{HPC12},   
a constraint-based model for conditional temporal planning. 
The \CSTN formalism extends {\STN}s in that:
(1) some of the nodes are called \emph{observation events} and to each of them 
is associated a boolean variable, to be disclosed only at execution time;
(2) \emph{labels} (i.e. conjunctions over the literals) 
are attached to all nodes \emph{and} constraints, to indicate the situations in which each of them is required. 
The planning agent must schedule all the required nodes, meanwhile respecting all the required temporal constraints among them. 
This extended framework allows for the off-line construction of conditional plans that are guaranteed to satisfy complex 
temporal constraints. Importantly, this can be achieved even while allowing for the decisions about the precise timing of actions 
to be postponed until execution time, in a least-commitment manner, thereby adding flexibility and making it possible to 
adapt the plan dynamically, during execution, in response to the observations made~\cite{TVP2003}.

Three notions of consistency arise for \CSTN{s}: weak, strong, and \emph{dynamic}. 
Dynamic consistency (DC) is in fact the most interesting one, as it requires the existence of conditional plans where 
decisions about the precise timing of actions are postponed until execution time, 
but it anyhow guarantees that all the relevant constraints will be ultimately satisfied. 
Still, it is the most challenging and it was conjectured to be hard to assess by Tsamardinos, Vidal and Pollack~\cite{TVP2003}. 
Indeed, the best-so-far algorithm for deciding whether a \CSTN is dynamically-consistent is doubly-exponential time~\cite{TVP2003}. 
It first builds an equivalent Disjunctive Temporal Problem (\DTP) of size exponential in the input \CSTN, 
and then applies to it an exponential time \DTP's algorithm to check its consistency. 
However, this approach turns out to be limitative in practice: to the best of our knowledge, 
some experimental studies have shown that the resolution procedures, as well as the heuristics, for solving general 
\DTP{s} becomes quite burdensome with $\sim30,35$ \DTP's variables~\cite{TsamardinosP03, MoffittP05, Oddi14}, 
thus dampening the practical applicability of the approach. 

\paragraph{Contribution}
In the present work we first offer a proof that deciding whether
a \CSTN is dynamically-consistent is \coNP-hard.
Secondly, and most importantly, we unveil a connection between the problem of checking dynamic-consistency 
of \CSTN{s} and that of determining \MPG{s}, 
thus providing the first sound-and-complete singly-exponential time algorithm for this same task of 
deciding the dynamic-consistency and yielding a dynamic execution strategy for \CSTN{s}. 
The algorithm can actually be applied to a wider class of problems and it is based on representing
any given instance on an exponential sized network, as first suggested in \cite{TVP2003}. 
The difference, however, is that we propose to map \CSTN{s} on \HTN{s}/\MPG{s} rather than on \DTP{s}. 
This makes a relevant difference since the consistency check for \HTN{s}
can be reduced to \MPG{s} determination~\cite{CPR2014, CPR2015}, 
which is amenable to practical and effective pseudo-polynomial time algorithms 
(indeed, in several cases the resolution methods for determining \MPG{s} 
exhibit even a strongly-polynomial time behaviour~\cite{CPR2015, BC12, AllamigeonBG14, ChatterjeeHKN14}). 
To summarize, we obtain an improved upper bound on the theoretical time complexity of the DC-checking for \CSTN{s} 
(\ie from $\text{2-EXP}$ to $\text{NE}\cap \text{coNE}$) together with a faster DC-checking procedure, 
which can be used on \CSTN{s} with a larger number of propositional variables and event nodes. 
At the heart of the algorithm a suitable reduction to \MPG{s} is mediated by the \HTN model, 
\ie the algorithm decides whether a \CSTN is dynamically-consistent by solving a carefully constructed \MPG.
As a final contribution, in order to analyze the algorithm, we introduce a novel and refined notion of dynamic-consistency, 
named $\epsilon$-dynamic-consistency, and present a sharp lower bounding analysis on the critical value of the 
\emph{reaction time} $\hat{\varepsilon}$ where the \CSTN transits from being, to not being, dynamically-consistent.
We believe that this contributes to clarifying  
(with respect to previous literature~\cite{TVP2003, HPC12})  
the role played by the reaction time $\hat{\varepsilon}$ 
in checking the dynamic-consistency of \CSTN{s}.
Furthermore, the proof technique introduced in this analysis of $\hat{\varepsilon}$
is applicable more in general when dealing with linear difference constraints which include strict inequalities, 
therefore, it may  be useful in the analysis of other models of temporal constraints.

\paragraph{Organization}
In Section~\ref{sect:backgroundandnotation}~A we recall the basic formalism, 
terminology and known results on \CSTP{s} and \CSTN{s}. Section~\ref{sect:backgroundandnotation}~B is devoted 
to recall the \HTN model, its computational equivalence with \MPG{s} and the related algorithmic results. 
Section~\ref{sect:Algo} tackles on the algorithmics of dynamic-consistency: firstly, we provide a $\coNP$-hardness lower bound, 
then, we describe the connection with \HTN{s}/\MPG{s} and present a (pseudo) singly-exponential time DC-checking procedure.
Section~\ref{sect:epsilon} is devoted to present a sharp lower bounding analysis on the critical value of the 
\emph{reaction time} $\hat{\varepsilon}$ where the \CSTN transits from being, to not being, dynamically-consistent.
In Section~\ref{sect:relatedworks} some related works are discussed.
The paper concludes in Section~\ref{sect:conclusions}. 

\section{Background}\label{sect:backgroundandnotation}
In order to provide a formal support to the present work, 
this section recalls the basic formalism, terminology and known results on \CSTP{s} and \CSTN{s}.
Since the forthcoming definitions are mostly inherited from the literature, 
the reader is referred to \cite{TVP2003} and \cite{HPC12} for an 
intuitive semantic discussion and for some clarifying examples of the very same model.

To begin with, our graphs are directed and weighted on the arcs.
Thus, if $G=\langle V,A\rangle$ is a graph, 
then every arc $a\in A$ is a triplet $\langle u,v,w_a\rangle$ where $u=t(a) \in V$ is 
the \textit{tail} of $a$, $v=h(a) \in V$ is the \textit{head} of $a$, and  
$w_a=w(u,v)\in\Z$ the \textit{weight} of $a$. 
The following definition recalls Simple Temporal Networks (\STN{s})~\cite{DechterMP91}, 
as they provide a powerful and general tool for representing conjunctions 
of minimum and maximum distance constraints between pairs of temporal variables. 
\begin{Def}[\STN{s}]
An \STN \cite{DechterMP91} is a weighted directed graph whose nodes are events that must be placed on the real time line and whose arcs, 
called \emph{standard arcs}, express binary constraints on the allocations of their end-points in time.

An \STN $G=\langle V, A\rangle$ is called \textit{consistent} if it admits a \emph{feasible scheduling}, 
\ie a scheduling $\phi: V\mapsto \RR$ such that $\phi(v) \leq \phi(u) + w(u,v)$ for all arcs $\langle u,v, w(u,v)\rangle\in A$.
\end{Def}
 
\subsection{Conditional Simple Temporal Networks}\label{subsect:CSTN}

In 2003, Tsamardinos, Vidal and Pollack introduced the \emph{Conditional Simple Temporal Problem (\CSTP)}  
as an extension of standard temporal constraint-satisfaction models used in non-conditional temporal planning.  
A \CSTP augments an \STN to include \emph{observation} events. Each observation event has a \emph{boolean variable} (or \emph{proposition}) 
associated with it. When the observation event is executed, the truth-value of its associated proposition becomes known. 
In addition, each event and each constraint has a \emph{label} that restricts the scenarios in which it plays a role.
Although not included in the formal definition, Tsamardinos, \etal discussed some supplementary \emph{reasonability assumptions} that any well-defined 
\CSTP must satisfy. Subsequently, those conditions have been analyzed and formalized in~\cite{HPC12}, 
leading to the sound notion of \emph{Conditional Simple Temporal Network} (\CSTN), 
which is now recalled.

Let $P$ be a set of boolean variables, a \emph{label} is any (possibly empty) conjunction of variables, or negations of variables, drawn from $P$.
The \emph{empty label} is denoted by $\lambda$. The \emph{label universe} of $P$, denoted $P^*$, is the set of all (possibly empty) labels whose literals 
are drawn from $P$. Two labels, $\ell_1$ and $\ell_2$, are called \emph{consistent}, 
denoted\footnote{The notation $\Con(\cdot, \cdot)$ and $\Sub(\cdot, \cdot)$ is inherited 
from~\cite{TVP2003, HPC12}.} by $\Con(\ell_1, \ell_2)$, when $\ell_1\wedge\ell_2$
is satisfiable. A label $\ell_1$ \emph{subsumes} a label $\ell_2$, denoted$^\text{1}$ by $\Sub(\ell_1, \ell_2)$, when $\ell_1\Rightarrow\ell_2$ holds. 
We are now in the position to recall the definition of {\CSTN}s. 
\begin{Def}[{\CSTN}s]
A \emph{Conditional Simple Temporal Network (\CSTN)} is a tuple $\langle V, A, L, \Ord, \Ord{V}, P \rangle$ where:
\begin{itemize}
\item $V$ is a finite set of \emph{events}; $P=\{p_1, \ldots, p_{|P|}\}$ is a finite set of \emph{boolean variables} (or \emph{propositions});
\item $A$ is a set of \emph{labeled temporal constraints} each having the form $\langle v-u\leq w(u,v), \ell\rangle$, 
where $u,v\in V$, $w(u,v) \in \Z$, and $\ell\in P^*$;
\item $L:V\rightarrow P^*$ is a function that assigns a label to each event in $V$; $\Ord{V}\subseteq V$ is a finite set of \emph{observation events};
$\Ord:P\rightarrow \Ord{V}$ is a bijection that associates a unique observation event $\Ord(p)=\Ord_p$ to each proposition $p\in P$;
\item The following \emph{reasonability assumptions} must hold:

(\emph{WD1}) for any labeled constraint $\langle v-u\leq w, \ell\rangle\in A$ the label $\ell$ is satisfiable 
and subsumes both $L(u)$ and $L(v)$; 
intuitively, whenever a constraint $v-u\leq w$ is required to be satisfied, 
then its endpoints $u$ and $v$ must be scheduled (sooner or later) by the planning agent;

(\emph{WD2}) for each $p\in P$ and each $u\in V$ such that either $p$ or $\neg p$ appears in $L(u)$, we require: 
$\Sub(L(u), L(\Ord_p))$, and $\langle \Ord_p-u\leq-\epsilon, L(u)\rangle \in A$ for some $\epsilon > 0$;
intuitively, whenever a label $L(u)$ of an event node $u$ contains proposition $p$, and $u$ gets eventually scheduled, 
then the observation event $\Ord_p$ must be scheduled strictly before $u$ by the planning agent.

(\emph{WD3}) for each labeled constraint $\langle v-u\leq w, \ell\rangle$ and $p\in P$, 
for which either $p$ or $\neg p$ appears in $\ell$, it holds that $\Sub(\ell, L(\Ord_p))$; 
intuitively, assuming a required constraint contains proposition $p$, 
then the observation event $\Ord_p$ must be scheduled (sooner or later) by the planner.

\end{itemize}
\end{Def}

\begin{Ex}
Fig.~\ref{FIG:cstn2} depicts an example of a \CSTN $\Gamma_1$ having three event nodes 
$A$, $B$ and $C$ as well as two observation events $\Ord_p$ and $\Ord_q$. 
\end{Ex}

\begin{figure}[!htb]
\centering
\begin{tikzpicture}[arrows=->,scale=.8,node distance=1 and 1]
 	\node[node] (A') {$A$};
	\node[node, xshift=15ex,right=of A'] (B') {$B$};
	\node[node, xshift=15ex,right=of B'] (C') {$C$};
	\node[node,below=of B', label={above:$p?$}] (P') {$\Ord_p$};
	\node[node,below=of P', label={above:$q?$}] (Q') {$\Ord_q$};
	\draw[] (A') to [bend left=40] node[timeLabel,above] {$10$} (C'); 
	\draw[] (C') to [bend right=30] node[timeLabel,below] {$-10$} (A'); 
	\draw[] (A') to [] node[timeLabel,above] {$3, p \neg q$} (B'); 
	\draw[] (B') to [] node[timeLabel,above] {$2, q$} (C'); 
	\draw[] (B') to [bend right=20] node[timeLabel,above] {$0$} (A'); 
	\draw[] (A') to [] node[xshift=1ex, yshift=0ex, timeLabel,above] {$5$} (P'); 
	\draw[] (P') to [bend left=15] node[xshift=1ex, yshift=0ex, timeLabel,above] {$0$} (A'); 
	\draw[] (A') to [] node[xshift=-1ex,yshift=0ex, timeLabel,below] {$9$} (Q'); 
	\draw[] (Q') to [bend left=15] node[xshift=-1ex,yshift=0ex, timeLabel,below] {$0$} (A'); 
	\draw[] (P') to [] node[xshift=-1ex, timeLabel,above] {$10$} (C'); 
	\draw[] (Q') to [] node[xshift=2ex, timeLabel,below] {$1, \neg p$} (C'); 
\end{tikzpicture}
\caption{
A \CSTN $\Gamma=\langle V, A, L, \Ord, \Ord{V}, P\rangle$ 
having two observation events $\Ord_p$ and $\Ord_q$. 
Formally, we have $V=\{A,B,C,\Ord_p, \Ord_q\}$, $P=\{p,q\}$, 
$\Ord{V}=\{\Ord_p, \Ord_q\}$, $L(v)=\lambda$ for every $v\in V$, 
$\Ord(p)=\Ord_p, \Ord(q)=\Ord_q$.
The set of labeled temporal constraints is:  
$A=\{ \langle C-A\leq 10, \lambda \rangle, \langle A-C\leq -10, 
\lambda \rangle, \langle B-A\leq 3, p\wedge\neg q \rangle, 
\langle A-B\leq 0, \lambda \rangle,
\langle \Ord_p-A\leq 5, \lambda \rangle, 
\langle A-\Ord_p\leq 0, \lambda \rangle,
\langle \Ord_q-A\leq 9, \lambda \rangle, 
\langle A-\Ord_q\leq 0, \lambda \rangle,
\langle C-B\leq 2, q \rangle, 
\langle C-\Ord_p\leq 10, \lambda\rangle,
\langle C-\Ord_q\leq 1, \neg p\rangle
\}$.}
\label{FIG:cstn2}
\end{figure}

In the following definitions we will implicitly 
refer to some \CSTN which is denoted 
$\Gamma=\langle V, A, L, \Ord, \Ord{V}, P \rangle$.
\begin{Def}[Scenario]
A \emph{scenario} over a set $P$ of boolean variables is a truth assignment $s:P\rightarrow \{\top, \bot\}$, \ie
$s$ is a function that assigns a truth value to each proposition $p\in P$. The set of all scenarios over $P$ is denoted $\Sigma_P$.
If $s\in\Sigma_P$ is a scenario and $\ell\in P^*$ is a label, 
then $s(\ell)\in\{\top, \bot\}$ denotes the truth value of $\ell$ induced by $s$ in the natural way.
\end{Def}
Notice that any scenario $s\in\Sigma_P$ can be described by means of the label
$\ell_s\triangleq l_1\wedge\cdots\wedge l_{|P|}$ such that,
for every $1\leq i\leq |P|$, the literal $l_i\in\{p_i, \neg p_i\}$ satisfies $s(l_i)=\top$.
\begin{Ex}
Consider the set of propositional variables $P=\{p,q\}$.
The scenario $s:P\rightarrow\{\top, \bot\}$ defined as $s(p)=\top$ and  
$s(q)=\bot$ can be compactly described by the label $\ell_s=p\wedge \neg q$.
\end{Ex}

\begin{Def}[Scheduling]
A \emph{scheduling} for a subset of events $U\subseteq V$ is a function $\phi:U\rightarrow\RR$ that assigns a real number to each
event in $U$. The set of all schedules over $U$ is denoted~$\Phi_U$. 
\end{Def}
\begin{Def}[Scenario Restriction]
Let $s\in\Sigma_{P}$ be a scenario. 
The \emph{restriction} of $V$ and $A$ \wrt $s$ are defined as follows:
\begin{itemize} 
\item $V^+_s\triangleq \{v\in V\mid s(L(v))=\top\}$; 
\item $A^+_s\triangleq \{\langle u,v,w\rangle \mid \exists {\ell}\, \langle v-u\leq w, \ell\rangle \in A, s(\ell)=\top\}$.
\end{itemize}

The restriction of $\Gamma$ \wrt $s$ is defined as $\Gamma^+_s\triangleq \langle V^+_s, A^+_s\rangle$. 
Finally, it is worth to introduce the notation $V^+_{s_1, s_2} \triangleq V^+_{s_1}\cap V^+_{s_2}$.
\end{Def}
We remark that the restriction $\Gamma^+_s$ is always an \STN. 

\begin{Def}[Execution Strategy]\label{def:executionstrategy}
An \emph{execution strategy} for $\Gamma$ is a mapping $\sigma:\Sigma_P\rightarrow \Phi_V$ such that, 
for any scenario $s\in\Sigma_P$, the domain of the scheduling $\sigma(s)$ is $V^+_{s}$. 
The set of execution strategies of $\Gamma$ is denoted by $\S_{\Gamma}$.
The \emph{execution time} of an event $v\in V^+_{s}$ in the schedule $\sigma(s)\in\Phi_{V^+_s}$ is denoted~by~$[\sigma(s)]_v$.
\end{Def}
\begin{Def}[Scenario History]\label{def:scenario_history}
Let $\sigma\in\S_{\Gamma}$ be an execution strategy, let $s\in\Sigma_P$ be a scenario and let $v\in V^+_{s}$ be an event.
The \emph{scenario history} $\text{scHst}(v,s,\sigma)$ of $v$ in the scenario $s$ for the strategy $\sigma$ is defined as:
$\text{scHst}(v,s,\sigma)\triangleq \{(p, s(p))\mid  p\in P,\, \Ord_p \in V^+_{s}\cap{\Ord}V,\, [\sigma(s)]_{\Ord_p} < [\sigma(s)]_v \}$.

\end{Def}
The scenario history can be compactly expressed by the conjunction 
of the literals corresponding to the observations comprising it. 
Thus, we may treat a scenario history as though it were a label.
\begin{Def}[Viable Execution Strategy]
We say that $\sigma\in\S_{\Gamma}$ is a \emph{viable} execution strategy if, for each scenario $s\in\Sigma_P$, 
the scheduling $\sigma(s)\in\Phi_V$ is feasible for the \STN $\Gamma^+_s$.
\end{Def}
\begin{Def}[Dynamic Consistency]\label{def:consistency}
An execution strategy $\sigma\in \S_{\Gamma}$ is called \emph{dynamic} if, 
for any $s_1, s_2\in \Sigma_P$ and any event $v\in V^+_{s_1, s_2}$, 
the following implication holds:
\[\Con(\scHst(v, s_1, \sigma), s_2) \Rightarrow [\sigma(s_1)]_v = [\sigma(s_2)]_v.\]
We say that $\Gamma$ is \emph{dynamically-consistent} 
if it admits $\sigma\in\S_{\Gamma}$ which is both viable and dynamic.
The problem of checking whether a given \CSTN is dynamically-consistent is named \emph{\DCC}.
\end{Def}

\begin{Ex}
Consider the \CSTN $\Gamma$ of \figref{FIG:cstn2}, 
and let the scenarios $s_1, s_2, s_3, s_4$ be defined as:
$s_1(p)=\bot$, $s_1(q)=\bot$; $s_2(p)=\bot$, $s_2(q)=\top$; 
$s_3(p)=\top$, $s_3(q)=\bot$; $s_4(p)=\top$, $s_4(q)=\top$. 
It follows an example of execution strategy $\sigma\in\S_\Gamma$:
$[\sigma(s_i)]_A=0$ for every $i\in\{1,2,3,4\}$;
$[\sigma(s_i)]_B=8$ for every $i\in\{1,2,4\}$ and $[\sigma(s_3)]_B=3$;
$[\sigma(s_i)]_C=10$ for every $i\in\{1,2,3,4\}$;
$[\sigma(s_i)]_{\Ord_p}=1$ for every $i\in\{1,2,3,4\}$.
$[\sigma(s_i)]_{\Ord_q}=2$ for every $i\in\{3,4\}$ 
and $[\sigma(s_i)]_{\Ord_q}=9$ for every $i\in\{1,2\}$.
The reader can check that $\sigma$ is viable and dynamic. 
Indeed, $\sigma$ admits the tree-like representation 
depicted in Fig~\ref{FIG:cstn2-strategy}.
\end{Ex}

\begin{figure}[!htb]
\centering
\begin{tikzpicture}[scale=0.9, level distance=50pt,sibling distance=30pt]
\Tree [. \framebox{$\phi(A)=0$}
\edge node[]{}; 
[. \framebox{$\phi(\Ord_p)=1$}
\edge node[left,xshift=-1ex]{$s(p)=\top$}; 
[. \framebox{$\phi(\Ord_q)=2$} 
\edge node[left,xshift=-1ex]{$s(q)=\top$}; 
[. \framebox{$\phi(B)=8$}
\edge node[]{}; \framebox{$\phi(C)=10$}
]
\edge node[right,xshift=1ex]{$s(q)=\bot$}; 
[. \framebox{$\phi(B)=3$}
\edge node[]{}; \framebox{$\phi(C)=10$} ] 
]
\edge node[right,xshift=1ex]{$s(p)=\bot$}; 
[. \framebox{$\phi(B)=8$}
\edge node[]{}; [.
 \framebox{$\phi(\Ord_q)=9$}
 \edge node[right]{$s(q)=\top,\bot$}; [. \framebox{$\phi(C)=10$} ] ]
 ]
]]
\end{tikzpicture}
\caption{A tree-like representation of a dynamic execution strategy $\sigma$ for the \CSTN $\Gamma$ of \figref{FIG:cstn2}, 
where $s$ denotes scenarios and $\phi(X)=[\sigma(s)]_X$ denotes the corresponding scheduling.}
\label{FIG:cstn2-strategy}
\end{figure}
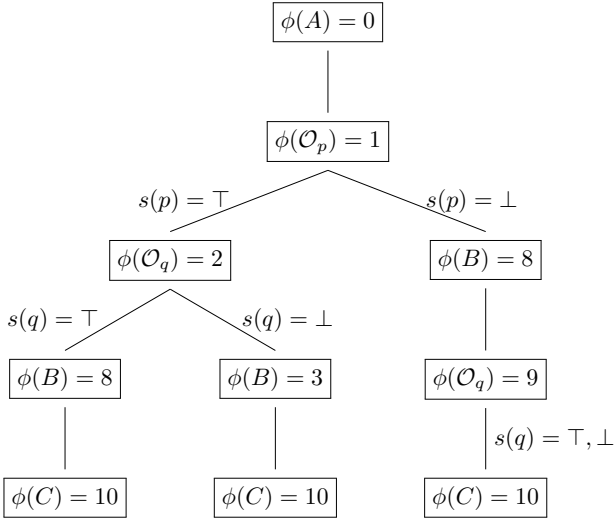

We introduce next a crucial notion for studying dynamic-consistency of \CSTN{s}, 
that is the \emph{difference set} $\Delta(s_1; s_2)$.
\begin{Def}[Difference Set]
Let $s_1, s_2\in\Sigma_P$ be two scenarios. 
The set of observation events in $V^+_{s_1}\cap{\Ord}V$ at which $s_1$ and $s_2$ differ is denoted by $\Delta(s_1;s_2)$. 
Formally, 
\[\Delta(s_1; s_2)\triangleq\{\Ord_p \in V^+_{s_1}\cap{\Ord}V \mid s_1(p)\neq s_2(p)\}.\]
\end{Def}
Notice that commutativity may not hold, \ie in general it may be the case that $\Delta(s_1; s_2)\neq \Delta(s_2; s_1)$.
\begin{Ex}
Consider the \CSTN $\Gamma$ of \figref{FIG:cstn2} and the scenarios $s_1, s_2$ defined:
$s_1(p)=\bot$, $s_1(q)=\bot$; $s_2(p)=\bot$, $s_2(q)=\top$.

Then, $\Delta(s_1;s_2)=\{\Ord_q\}$. 
\end{Ex}
The next lemma will be useful later on in Section~\ref{sect:Algo}.
\begin{Lem}\label{lem:dynamicimplequality}
Let $s_1, s_2\in \Sigma_P$ and  $v\in V^+_{s_1, s_2}$.
Let $\sigma\in\S_{\Gamma}$ be an execution strategy. 
Then, $\sigma$ is dynamic if and only if the following implication holds for every $s_1, s_2\in\Sigma_P$, $u\in V^+_{s_1, s_2}$:
\[\Big(\bigwedge_{v\in\Delta(s_1; s_2)} [\sigma(s_1)]_u \leq [\sigma(s_1)]_v\Big) 
\Rightarrow [\sigma(s_1)]_u = [\sigma(s_2)]_u\;\; (L\ref{lem:dynamicimplequality})\] 
\end{Lem}
\begin{proof}
Notice that, by definition of $\Con(\cdot, \cdot)$ and $\scHst(\cdot, \cdot, \cdot)$, $\Con(\scHst(u, s_1, \sigma), s_2)$ holds
if and only if there is no observation event $v\in\Delta(s_1; s_2)$ which is scheduled by $\sigma(s_1)$ strictly before $u$. 
Therefore, $\Con(\scHst(u, s_1, \sigma), s_2)$ holds if and only if $\bigwedge_{v\in\Delta(s_1; s_2)} [\sigma(s_1)]_u \leq [\sigma(s_1)]_v$.
Thus, by definition of dynamic execution strategy, the thesis follows.
\end{proof}

\subsection{Hyper Temporal Networks}\label{subsect:HTN}

This subsection surveys the \textit{Hyper Temporal Network} (\HTN) model, 
which is a strict generalization of \STN{s}. The reader is referred to~\cite{CPR2014, CPR2015} for an in-depth treatise on \HTN{s}.

\begin{Def}[Hypergraph]
A \emph{hypergraph} $\H$ is a pair $\langle V,\A\rangle$, where $V$ is the set of nodes, and $\A$ is the set of \emph{hyperarcs}.
Each hyperarc $A=\langle t_A, H_A, w_A\rangle\in \A$ has a distinguished node $t_A$, called the \emph{tail} of $A$, and a nonempty set
$H_A\subseteq V\setminus\{t_A\}$ containing the \emph{heads} of $A$; 
to each head $v\in H_A$ is associated a \emph{weight} $w_A(v)\in\Z$. 
\end{Def}

Provided that $|A| \triangleq |H_A\cup \{t_A\}|$, the \emph{size} of a hypergraph $\H = \langle V,\A\rangle$ is defined as $m_{\A}\triangleq \sum_{A\in\A}|A|$, 
and it is used as a measure for the encoding length of $\H$.
If $|A|=2$, then $A=\langle u, v, w \rangle$ can be regarded as a \emph{standard arc}.
In this way, hypergraphs generalize graphs.

A \HTN is a weighted hypergraph $\H=\langle V,\A\rangle$ where a node represents an \emph{event} to be scheduled, 
and a hyperarc represents a set of temporal distance \emph{constraints} between the \emph{tail} and the \emph{heads}, 

In the \HTN framework the consistency problem is defined to be the following decision problem.
\begin{Def}[\HTNC] 
Given a \HTN \mbox{$\H=\langle V,\A\rangle$}, 
decide whether there exists a scheduling function \mbox{$\phi:V \rightarrow \RR$} such that:
\[\phi(t_A) \geq \min_{v\in H_A} \phi(v) - w_A(v), \;\forall\; A\in\A\]
any such scheduling \mbox{$\phi:V \rightarrow \RR$} is called \textit{feasible}.

A \HTN is called \textit{consistent} whenever it admits at least one feasible scheduling.
The problem of checking whether a given \HTN is consistent is named \HTNC.
\end{Def}

Indeed, observe that \HTNC generalizes \STN-Consistency because an \STN may be viewed as a \HTN.
The converse is not true because feasible schedules for a \HTN do not need to form a convex 
polytope~\cite{CPR2014, CPR2015} whereas, in general, the feasible schedules of an \STN are the solutions 
of a linear system and, therefore, they form a convex polytope.

The computational equivalence between checking the consistency of {\HTN}s 
and determining the winning regions of {\MPG}s was pointed out in~\cite{CPR2014, CPR2015}. 
The tightest worst-case time complexity for solving \HTNC is expressed by the following theorem, 
which was proven by resorting to the Value Iteration Algorithm for \MPG{s}~\cite{brim2011faster}.
The approach was shown to be robust by experimental evaluations in~\cite{CPR2015, BC12}, 
where \HTN{s} of size $\sim 10^6$ were solved within $\sim 5$ sec.

\begin{Thm}{\cite{CPR2014}}\label{Teo:MainAlgorithms}
The following propositions hold on \HTN{s}. 
\begin{enumerate}
\item There exists an $O((|V|+|\A|) m_{\A} W)$ pseudo-polynomial time algorithm for checking \HTNC;
\item \label{Cor:PseudoPolyScheduling}
There exists an $O((|V|+|\A|) m_{\A} W)$ pseudo-polynomial time algorithm such that, 
given in input any consistent \HTN $\H=(V, \A)$, then it returns as output a feasible scheduling $\phi:V\rightarrow \RR$ of $\H$;
\end{enumerate}
Here, $W\triangleq \max_{A\in\A, v\in H_A} |w_A(v)|$.
\end{Thm}

\section{Algorithmics of Dynamic-Consistency}\label{sect:Algo}

To start with, we offer the following $\coNP$-hardness lower bound on \DCC.
\begin{Thm} 
\DCC is $\coNP$-hard.
\end{Thm}
\begin{proof} 
We reduce $3$-\SAT\, to the complement of \DCC. Let $\varphi$ be a boolean formula in 3CNF. 
Let $X$ be the set of variables and let $\C=\{C_0, \ldots, C_{m-1}\}$ 
be the set of clauses comprising $\varphi = \bigwedge_{j=0}^{m-1} C_j$. 

(1) Let $N^\varphi$ be the \CSTN $\langle V^\varphi, A^\varphi, L^\varphi, \Ord^\varphi, \Ord{V}^\varphi, P^\varphi \rangle$, 
where: $V^{\varphi}\triangleq X\cup \C $, and all the nodes are given empty label, \ie $L^{\varphi}(v)=\lambda$ for every $v\in V^{\varphi}$; 
$P^{\varphi}\triangleq \Ord{V^{\varphi}} \triangleq X$; $\Ord^{\varphi}$ is the identity function;
for every $u,v\in\Ord{V^{\varphi}}$ we have $\langle u-v\leq 0, \lambda\rangle \in A^{\varphi}$; 
for every $x\in X$ and $C\in \C$ we have $\langle x-C\leq -1, \lambda\rangle \in A^{\varphi}$;
for each $j=0,\ldots, m-1$ and for each literal $\ell\in C_j$, 
we have $\langle C_j-C_{(j+1)\text{mod } m}\leq -1, \ell\rangle \in \A_\varphi$. 
Notice that $|V^{\varphi}| = n+m$ and $|A^{\varphi}|=n^2+nm+3m$.

(2) Assume that $\varphi$ is satisfiable. Let 
$\nu$ be a satisfying truth-assignment of $\varphi$.
In order to prove that $N^\varphi$ is not dynamically-consistent, 
observe that the restriction of $N^\varphi$ \wrt the scenario $\nu$ is a non-consistent \STN.
Indeed, if for every $j=0, \ldots, m-1$ we pick a standard arc $\langle C_j-C_{(j+1)\text{mod } m}\leq-1,\ell_j\rangle$ 
with $\ell_j$ being a literal in $C_j$ such that $\nu(\ell_j)=\top$, then we obtain a negative circuit.

(3) Assume that $\varphi$ is unsatisfiable. In order to prove that $N^{\varphi}$ is dynamically-consistent, 
we exhibit a viable and dynamic execution strategy $\sigma$ for $N^{\varphi}$.
First, schedule every $x\in X$ at $\sigma(x)\triangleq 0$. 
Therefore, by time $1$, the planner has full knowledge of the observed scenario $\nu$.
Since $\varphi$ is unsatisfiable, there exists an index $j_{\nu}$ such that $\nu(C_{j_\nu})=\bot$.
At this point, set $\sigma(C_{(j_{\nu}+k)\text{mod }m})\triangleq k$ for $k=1, \ldots, m$.
The reader can verify that $\sigma$ is viable and dynamic for $N^{\varphi}$.
\end{proof}

It remains currently open whether \DCC lies in {\PSPACE} and whether it is \PSPACE-hard.

\subsection{$\epsilon$-Dynamic-Consistency}
In \CSTN{s}, decisions about the precise timing of actions are postponed until execution time, 
when informations meanwhile gathered at the observation nodes can be taken into account.
However, the planner is allowed to factor in an outcome, and differentiate its strategy according to it, 
only strictly after the outcome has been observed (whence the strict inequality in Definition~\ref{def:scenario_history}).
Notice that this definition does not take into account the reaction time, which, in most applications, is non-negligible. 
In order to deliver algorithms that can also deal with the \emph{reaction time} $\epsilon$ of the planner, 
we employ a refined notion of dynamic-consistency.

\begin{Def}[$\epsilon$-dynamic-consistency]\label{def:epsilonconsistency}
Given any \CSTN $\langle V, A, L, \Ord, \Ord{V}, P \rangle$ and any real number $\epsilon\in (0, +\infty)$, 
an execution strategy $\sigma\in\S_{\Gamma}$ is \emph{$\epsilon$-dynamic} if it satisfies all the $H_\epsilon\text{-constraints}$,  
namely, for any two scenarios $s_1, s_2\in \Sigma_P$ and any event $u\in V^+_{s_1, s_2}$, 
the execution strategy $\sigma$ satisfies the following constraint, which is denoted $H_{\epsilon}(s_1;s_2;u)$:
\[
[\sigma(s_1)]_u \geq \min\Big(\{[\sigma(s_2)]_u\}\cup\{[\sigma(s_1)]_v + \epsilon\mid v\in\Delta(s_1; s_2)\}\Big)
\]
We say that a \CSTN $\Gamma$ is \emph{$\epsilon$-dynamically-consistent} if it admits $\sigma\in\S_{\Gamma}$ 
which is both viable and $\epsilon$-dynamic. 
The problem of checking whether a given \CSTN is $\epsilon$-dynamically-consistent is named \emph{\eDCC}.
\end{Def}

It follows directly from Definition~\ref{def:epsilonconsistency} that,  
whenever $\sigma\in S_{\Gamma}$ satisfies some $H_\epsilon(s_1;s_2;u)$, then $\sigma$ satisfies $H_{\epsilon'}(s_1;s_2;u)$ 
for every $\epsilon’\in(0, \epsilon]$ as well. This proves the following lemma.
\begin{Lem}
If $\Gamma$ is $\epsilon$-dynamically-consistent, for some $\epsilon>0$, 
then $\Gamma$ is $\epsilon'$-dynamically-consistent for every $\epsilon'\in (0, \epsilon]$.
\end{Lem}
Given any dynamically-consistent \CSTN, 
we may ask for the maximum reaction time $\epsilon$ of the planner beyond which the network is no longer dynamically-consistent. 
\begin{Def}[Reaction time $\hat{\epsilon}$]
Let $\hat{\epsilon}=\hat{\epsilon}(\Gamma)$ be the greatest real number $\epsilon$ 
such that $\Gamma$ is $\epsilon$-dynamically-consistent. 
\end{Def}
If $\Gamma$ is dynamically-consistent, 
then $\hat{\epsilon}(\Gamma)$ exists finite and $\hat{\epsilon}(\Gamma)>0$, 
as it is now proved in Lemma~\ref{lem:dynamic_impl_epsilon}.
\begin{Lem}\label{lem:dynamic_impl_epsilon}
Let $\sigma$ be a dynamic execution strategy for the \CSTN $\Gamma$.
Then, there exists a sufficiently small real number $\epsilon\in (0, +\infty)$ such that $\sigma$ is $\epsilon$-dynamic.
\end{Lem}
\begin{proof} 
Let $s_1, s_2\in \Sigma_P$ be two scenarios and let us consider any event $u\in V^+_{s_1, s_2}$.
Since $\sigma$ is dynamic, then by Lemma~\ref{lem:dynamicimplequality} the following implication necessarily holds:
\[\Big( \bigwedge_{v\in\Delta(s_1; s_2)} [\sigma(s_1)]_u \leq [\sigma(s_1)]_v \Big) \Rightarrow [\sigma(s_1)]_u \geq [\sigma(s_2)]_u\;\; (*)\] 
Notice that, \wrt Lemma~\ref{lem:dynamicimplequality}, 
we have relaxed the equality $[\sigma(s_1)]_u = [\sigma(s_2)]_u$ in the implicand of (L\ref{lem:dynamicimplequality}) 
by introducing the inequality $[\sigma(s_1)]_u \geq [\sigma(s_2)]_u$.
At this point, we convert ($*$) from implicative to disjunctive form, 
first by applying the rule of material implication\footnote{The rule of material implication: $\models p\Rightarrow q \iff \neg p \vee q$}, 
and then De Morgan's law\footnote{De Morgan's law: $\models \neg (p\wedge q)\iff \neg p\vee \neg q$}. 
From this, we see that the following disjunction must hold: 
\[
 [\sigma(s_1)]_u \geq [\sigma(s_2)]_u 
\vee \bigvee_{v\in\Delta(s_1; s_2)} [\sigma(s_1)]_u > [\sigma(s_1)]_v\;\; (**) 
\]
Then, we argue that there exists a real number $\epsilon\in (0, +\infty)$ 
such that the following disjunction holds as well:
\[
[\sigma(s_1)]_u \geq [\sigma(s_2)]_u   
\vee \bigvee_{v\in\Delta(s_1; s_2)} [\sigma(s_1)]_u \geq [\sigma(s_1)]_v +\epsilon .
\]
In fact, since the disjunction ($**$) necessarily holds, then  
one can define $\epsilon$ to be the minimum among all the values $\epsilon(s_1;s_2;u)\in (0, +\infty)$ 
such that for every $\langle s_1, s_2, u\rangle\in \Sigma_P\times\Sigma_P\times V^+_{s_1, s_2}$ 
the following is satisfied:
$\epsilon(s_1;s_2;u) = 1$ if $[\sigma(s_1)]_u \geq [\sigma(s_2)]_u$; 
otherwise, $\epsilon(s_1;s_2;u) = \min\{[\sigma(s_1)]_u - 
[\sigma(s_1)]_v \mid v\in\Delta(s_1, s_2),\; [\sigma(s_1)]_u > [\sigma(s_1)]_v\}$.

This implies that $\sigma$ satisfies every $H_{\epsilon}$-constraint of $\Gamma$, 
and thus that $\sigma$ is $\epsilon$-dynamic.
\end{proof}
\begin{Lem}\label{lem:epsilon_impl_dynamic}
Let $\sigma$ be an $\epsilon$-dynamic execution strategy for the \CSTN $\Gamma$, 
for some $\epsilon\in (0, +\infty)$.
Then, $\sigma$ is dynamic.
\end{Lem}
\begin{proof}
For the sake of contradiction, let us suppose that $\sigma$ is not dynamic.
Let $F$ be the non-empty set of all the triplets $\langle u,s_1,s_2\rangle\in V^+_{s_1, s_2}\times\Sigma_P\times\Sigma_P$, 
for which the implication (L\ref{lem:dynamicimplequality}) does not hold.
Then, $\langle u,s_1,s_2\rangle\in F$ if and only if the following two hold:
\begin{enumerate}
\item $[\sigma(s_1)]_u \leq [\sigma(s_1)]_v$ for every $v\in\Delta(s_1; s_2)$; 
\item $[\sigma(s_1)]_u \neq [\sigma(s_2)]_u$.
\end{enumerate}
Let $\langle\hat{u}, \hat{s_1}\rangle\triangleq \arg\min\{[\sigma(s_1)]_u\mid \exists{s_2}\, 
\langle u,s_1,s_2\rangle\in F\}$ be an event whose scheduling time is minimum and for which ($1$) and ($2$) hold.
Since $\langle \hat{u}, \hat{s_1}\rangle$ is minimum in time, then 
$[\sigma(\hat{s_1})]_{\hat{u}} \leq [\sigma(s_2)]_{\hat{u}}$ for every $s_2\in\Sigma_P$ such that $\langle\hat{u}, \hat{s_1}, s_2\rangle\in F$;
moreover, since $\langle\hat{u}, \hat{s_1}, s_2\rangle\in F$, then $[\sigma(s_1)]_{\hat{u}} \neq [\sigma(s_2)]_{\hat{u}}$ by (2), 
so that $[\sigma(\hat{s_1})]_{\hat{u}} < [\sigma(s_2)]_{\hat{u}}$. 
At this point, recall that $\sigma$ is $\epsilon$-dynamic by hypothesis, hence $[\sigma(\hat{s_1})]_{\hat{u}} < [\sigma(s_2)]_{\hat{u}}$ implies that 
there exists $v\in\Delta(\hat{s_1}; s_2)$ such that $[\sigma(\hat{s_1})]_{\hat{u}}\geq [\sigma(\hat{s_1})]_v+\epsilon>[\sigma(\hat{s_1})]_v$, 
but this inequality contradicts ($1$). 
Indeed, $F=\emptyset$ and $\sigma$ is thus dynamic.
\end{proof}
In Section~\ref{sect:epsilon}, the following theorem is proved.
\begin{Thm}\label{thm:boundexp} 
For any dynamically-consistent \CSTN $\Gamma$, 
where $V$ is the set of events and $\Sigma_P$ is the set of scenarios, 
we have that $\hat{\epsilon}(\Gamma) \geq |\Sigma_P|^{-1}|V|^{-1}$.
\end{Thm}
Notice that, in Definition~\ref{def:consistency}, dynamic-consistency was defined by strict-inequality and equality constraints.
However, by Theorem~\ref{thm:epsilonconsistency}, 
dynamic-consistency can also be defined in terms of $H_{\epsilon}$-constraints only 
(\ie no strict-inequalities are required).
\begin{Thm}\label{thm:epsilonconsistency}
Let $\epsilon\triangleq |\Sigma_P|^{-1}|V|^{-1}$. 
Then, $\Gamma$ is dynamically-consistent if and only if $\Gamma$ is $\epsilon$-dynamically-consistent.
\end{Thm}

By Theorem~\ref{thm:epsilonconsistency}, 
any algorithm for checking $\epsilon$-dynamic-consistency can be used to check dynamic-consistency as well. 

\subsection{A Singly-Exponential Time Algorithm for \DCC}\label{subsect:Algo}

In this section, we present the first singly-exponential time algorithm for solving \DCC, 
also producing a dynamic execution strategy whenever the input \CSTN is dynamically-consistent. 
Hereafter, let us denote $\N_0\triangleq \N\setminus\{0\}$.
The main result of this paper is summarized in the following theorem, which is proven in this section.
\begin{Thm}\label{thm:mainresult} 
The following two propositions hold.
\begin{enumerate}
\item There exists an 
$O(|\Sigma_P|^{2}|A|^2 +|\Sigma_P|^3|A||V||P| + |\Sigma_P|^4|V|^2|P|)WD$ time algorithm 
deciding \eDCC on input $\langle\Gamma, \epsilon\rangle$, 
for any \CSTN $\Gamma=\langle V, A, L, \Ord, \Ord{V}, P \rangle$ and any rational number $\epsilon=N/D$ where $N,D\in \N_0$.
In particular, given any $\epsilon$-dynamically-consistent \CSTN $\Gamma$, 
the algorithm returns as output a viable and $\epsilon$-dynamic execution strategy for $\Gamma$.
\item There exists an $O(|\Sigma_P|^{3}|A|^2|V| + |\Sigma_P|^4|A||V|^2|P| + |\Sigma_P|^5|V|^3|P|)W$ time algorithm for checking \DCC  
on any input $\Gamma=\langle V, A, L, \Ord, \Ord{V}, P \rangle$.
In particular, given any dynamically-consistent \CSTN $\Gamma$, 
the algorithm returns a viable and dynamic execution strategy for $\Gamma$.
\end{enumerate}
Here, $W\triangleq \max_{a\in A} |w_a|$.
\end{Thm}

We now present the reduction from \DCC to \HTNC. 

Firstly, we argue that any \CSTN can be viewed as a succinct representation 
which can be expanded into an exponential sized \STN.
The \emph{Expansion} of \CSTN{s} is introduced below.

\begin{Def}[Expansion $\langle V^{\text{Ex}}_{\Gamma}, \Lambda^{\text{Ex}}_{\Gamma}\rangle$]\label{def:expansion}
Let $\Gamma$ be a \CSTN $\langle V, A, L, \Ord, {\Ord}V, P \rangle$.
Consider the distinct \STN{s} $\langle V_s, A_s\rangle$,  
one for each scenario $s\in\Sigma_P$, defined as follows: 
\[V_s\triangleq\{v_s \mid v\in V^+_s\} \text{ and } A_s\triangleq\{\langle u_s, v_s, w\rangle \mid \langle u,v,w\rangle \in A^+_s\}.\]
We define the \emph{expansion} $\langle V^{\text{Ex}}_{\Gamma}, \Lambda^{\text{Ex}}_{\Gamma}\rangle$ of $\Gamma$ as follows:
\[\langle V^{\text{Ex}}_{\Gamma}, \Lambda^{\text{Ex}}_{\Gamma}\rangle\triangleq 
	\Big\langle\bigcup_{s\in\Sigma_P}V_s, \bigcup_{s\in\Sigma_P} A_s\Big\rangle.\] 
\end{Def}

Notice that $V_{s_1}\cap V_{s_2}=\emptyset$ whenever $s_1\neq s_2$ and that 
$\langle V^{\text{Ex}}_{\Gamma}, \Lambda^{\text{Ex}}_{\Gamma}\rangle$ is an \STN with at most $|V^{\text{Ex}}_{\Gamma}|\leq|\Sigma_P|\, |V|$ nodes 
and at most $|\Lambda^{\text{Ex}}_{\Gamma}|\leq |\Sigma_P|\, |A|$ standard arcs.

We now show that the expansion of a \CSTN can be enriched with some hyperarcs 
in order to model $\epsilon$-dynamic-consistency, by means of a particular \HTN which is denoted $\H_{\epsilon}(\Gamma)$.
\begin{Def}[\HTN $\H_{\epsilon}(\Gamma)$]\label{def:Hepsilonzero}
Given any $\epsilon\in(0, +\infty)$ and any \CSTN $\Gamma = \langle V, A, L, \Ord, {\Ord}V, P \rangle$,  
a corresponding \HTN denoted by $\H_{\epsilon}(\Gamma)$ can be defined as follows:
\begin{itemize}
\item For every scenarios $s_1, s_2\in\Sigma_P$ and every event $u\in V^+_{s_1, s_2}$,   
define a hyperarc $\alpha=\alpha_{\epsilon}(s_1; s_2; u)$ as follows 
(with the intention to model $H_{\epsilon}(s_1; s_2; u)$, see Def.~\ref{def:epsilonconsistency}):
\[\alpha_{\epsilon}(s_1; s_2; u)\triangleq \langle t_\alpha, H_\alpha, w_\alpha\rangle, \]
where: 
\begin{itemize}
\item $t_\alpha\triangleq u_{s_1}$ is the tail of the hyperarc $\alpha$; 
\item $H_\alpha\triangleq \{u_{s_2}\}\cup \Delta(s_1; s_2)$ is the set of the heads;   
\item $w_\alpha(u_{s_2})\triangleq 0$; $w_\alpha(v)\triangleq - \epsilon$ for each $v\in\Delta(s_1; s_2)$.  
\end{itemize}
\item Consider the expansion $\langle V^{\text{Ex}}_{\Gamma}, \Lambda^{\text{Ex}}_{\Gamma}\rangle$ of $\Gamma$. 
Then, $\H_{\epsilon}(\Gamma)$ is defined as  
$\H_{\epsilon}(\Gamma)\triangleq \langle V^{\text{Ex}}_{\Gamma}, \A_{H_{\epsilon}}\rangle$, 
where,  
\[\A_{H_{\epsilon}}\triangleq \Lambda^{\text{Ex}}_{\Gamma} \cup \bigcup_{\substack{s_1,s_2\in\Sigma_P \\ u\in V^+_{s_1, s_2}}}\alpha_{\epsilon}(s_1;s_2;u).\]
\end{itemize}
\end{Def}

Notice that each $\alpha_{\epsilon}(s_1; s_2; u)$ has size $|\alpha_{\epsilon}(s_1; s_2; u)| = O(\Delta(s_1;s_2))=O(|P|)$.
In Fig.~\ref{FIG:PseudocodeReduction-cstn-htn}, Algorithm~\ref{algo:pseudocode_construct_H} presents the pseudocode for constructing $\H_{\epsilon}(\Gamma)$.
An excerpt of the \HTN corresponding to the \CSTN of Fig.~\ref{FIG:cstn2} is depicted in \figref{FIG:cstn22htn}.

The following theorem establishes the connection between dynamic-consistency of \CSTN{s} and consistency of \HTN{s}.
\begin{Thm}\label{thm:mainreduction}
Given any \CSTN $\Gamma=\langle V, A, L, \Ord, {\Ord}V, P \rangle$, 
there exists a sufficiently small real number $\epsilon\in (0, +\infty)$ such that 
$\Gamma$ is dynamically-consistent if and only if $\H_{\epsilon}(\Gamma)$ is consistent.
Moreover, $\H_{\epsilon}(\Gamma)$ has at most 
$|V_{\H_{\epsilon}}|\leq |\Sigma_P|\, |V|$ nodes, $|\A_{\H_{\epsilon}}|=O(|\Sigma_P|\,|A| + |\Sigma_P|^2|V|)$ hyperarcs, 
and it has size at most $m_{\A_{\H_{\epsilon}}}=O(|\Sigma_P|\, |A| + |\Sigma_P|^2|V|\, |P|)$.
\end{Thm}
\begin{proof}
For any $\epsilon>0$, let $\H_{\epsilon}(\Gamma)=\langle V^{\text{Ex}}_{\Gamma}, \A_{H_{\epsilon}} \rangle$ be the \HTN of Definition~\ref{def:Hepsilonzero}.

(1) Firstly, we prove that, for any $\epsilon>0$, $\H_{\epsilon}(\Gamma)$ is consistent 
if and only if $\Gamma$ is $\epsilon$-dynamically-consistent.
($\Rightarrow$) Given any feasible scheduling $\phi:V_\Gamma^{\text{Ex}}\rightarrow\RR$ for $\H_{\epsilon}(\Gamma)$, 
let $\sigma_\phi(s)\in\S_{\Gamma}$ be the execution strategy 
defined as: $[\sigma_\phi(s)]_v\triangleq \phi(v_s)$, for every $v_s\in V_{\Gamma}^{\text{E}}$, where $v\in V$ and $s\in\Sigma_P$.
Notice that each hyperarc $\alpha_{\epsilon}(s_1; s_2; u)$ is satisfied by $\phi$  
if and only if the corresponding $H_{\epsilon}$-constraint $H_{\epsilon}(s_1; s_2; u)$ is satisfied by $\sigma_{\phi}$; 
moreover, recall that $\Lambda^{\text{Ex}}_{\Gamma}\subseteq \A_{H_{\epsilon}}$, 
and that $\Lambda^{\text{Ex}}_{\Gamma}$ contains all the original standard difference constraints of $\Gamma$. 
At this point, since $\phi$ is feasible for $\H_{\epsilon}(\Gamma)$, then $\sigma_{\phi}$ must be viable and $\epsilon$-dynamic for $\Gamma$.
Hence, $\Gamma$ is $\epsilon$-dynamically-consistent. 

($\Leftarrow$) Given any viable and $\epsilon$-dynamic execution strategy $\sigma\in\S_{\Gamma}$, for some $\epsilon>0$, 
let $\phi_{\sigma}:V^{\text{Ex}}_{\Gamma}\rightarrow\RR$ be the scheduling of $\H_{\epsilon}(\Gamma)$ 
defined as: $\phi_{\sigma}(v_s)\triangleq [\sigma(s)]_v$ for every $v_s\in V^{\text{Ex}}_{\Gamma}$, where $v\in V$ and $s\in\Sigma_P$.
Also in this case we have $\Lambda^{\text{Ex}}_{\Gamma}\subseteq \A_{H_{\epsilon}}$, 
and a moment's reflection reveals that each hyperarc $\alpha_{\epsilon}(s_1; s_2; u)$ is satisfied by $\phi_{\sigma}$  
if and only if  $H_{\epsilon}(s_1; s_2; u)$ is satisfied by $\sigma$.
At this point, since $\sigma$ is viable and $\epsilon$-dynamic for $\Gamma$, then $\phi_{\sigma}$ must be feasible for $\H_{\epsilon}(\Gamma)$.
Hence $\H_{\epsilon}(\Gamma)$ is consistent.

(2) At this point, by composition with (1), Lemma~\ref{lem:dynamic_impl_epsilon} implies that  
there exists a sufficiently small $\epsilon>0$ such that 
$\Gamma$ is dynamically-consistent if and only if $\H_{\epsilon}(\Gamma)$ is consistent.

(3) The size bounds follow directly from Definition~\ref{def:Hepsilonzero}. 
\end{proof}

\begin{figure}[!htb]
\removelatexerror
\begin{algorithm}[H]\label{algo:pseudocode_construct_H}
\caption{$\texttt{construct\_}\H$$(\Gamma, \epsilon)$}

\KwIn{a \CSTN $\Gamma\triangleq\langle V, A, L, \Ord, {\Ord}V, P \rangle$, 
a rational $\epsilon>0$}

\ForEach{($s\in\Sigma_P$)}{

$V_s\leftarrow\{v_s \mid v\in V^+_{s}\}$\;
 
$A_s\leftarrow\{a_s\mid a\in A^+_s\}$\;

}

$\displaystyle V^{\text{Ex}}_{\Gamma}\leftarrow \cup_{\substack{s\in\Sigma_P}}V_s$\;

$\displaystyle \Lambda^{\text{Ex}}_{\Gamma}\leftarrow \cup_{\substack{s\in\Sigma_P}} A_s$\;

\ForEach{($s_1, s_2\in\Sigma_P$ \texttt{AND} $u\in V^+_{s_1, s_2}$)}{

$t_\alpha\leftarrow u_{s_1}$\;

$H_\alpha\leftarrow \{u_{s_2}\}\cup \Delta(s_1; s_2)$\;

$w_\alpha(u_{s_2})\leftarrow 0$\;

\ForEach{$v\in\Delta(s_1; s_2)$}{

$w_\alpha(v_{s_1})\leftarrow -\epsilon$\;

}

$\alpha_{\epsilon}(s_1; s_2; u)\leftarrow \langle t_\alpha, H_\alpha, w_\alpha\rangle$\;

}

$\displaystyle\A_{\H_{\epsilon}}\leftarrow \Lambda^{\text{Ex}}_{\Gamma} \cup 
	\bigcup_{\substack{s_1,s_2\in\Sigma_P \\ u\in V^+_{s_1, s_2}}}\alpha_{\epsilon}(s_1;s_2;u)$\;

$\H_{\epsilon}(\Gamma)\leftarrow \langle V^{\text{Ex}}_{\Gamma}, \A_{H_{\epsilon}}\rangle$\;

\Return{$H_{\epsilon}(\Gamma)$;}
\end{algorithm}

\removelatexerror
\begin{algorithm}[H]\label{algo:solve_epsilonzeroDCC}
\caption{\texttt{check\_CSTN-$\epsilon$-DC}$(\Gamma, \epsilon)$}

\KwIn{a \CSTN $\Gamma\triangleq\langle V, A, L, \Ord, {\Ord}V, P \rangle$, 
a rational number $\epsilon\triangleq N/D$, for $N,D\in\N_0$}

$\H_{\epsilon}(\Gamma)\leftarrow\texttt{construct\_}\H(\Gamma, \epsilon)$; \tcp{ref. Algorithm~\ref{algo:pseudocode_construct_H}}

\ForEach{($A=\langle t_A, H_A, w_A\rangle\in\A_{\H_{\epsilon}(\Gamma)}$ \texttt{AND} $h\in H_A$)}{
	$w_A(h)\leftarrow w_A(h) D$; \tcp{scale weights to $\Z$}
}

$\phi\leftarrow \texttt{check\_\HTNC}(\H_{\epsilon}(\Gamma))$; \tcp{ref. Thm~\ref{Teo:MainAlgorithms}}

\If{($\phi$ is a feasible scheduling of $\H_{\epsilon}(\Gamma)$)}{

\ForEach{(event node $v\in V_{\H_{\epsilon}(\Gamma)}$)}{
	$\phi(v)\leftarrow \phi(v)/ D$; \tcp{re-scale back to size w.r.t $\epsilon$}
}

\Return{$\langle\texttt{YES}, \phi\rangle$;}

}
\Else{

\Return{$\texttt{NO}$;}

}
\end{algorithm}
\begin{algorithm}[H]\label{algo:solve_DCC}
\caption{\texttt{check\_\DCC}$(\Gamma)$}

\KwIn{a \CSTN $\Gamma\triangleq\langle V, A, L, \Ord, {\Ord}V, P \rangle$}

$\hat{\epsilon} \leftarrow |\Sigma_P|^{-1}|V|^{-1}$; \tcp{ref. Thm.~\ref{thm:epsilonconsistency}}

\Return{\texttt{check\_CSTN-$\epsilon$-DC$(\Gamma, \hat{\epsilon})$};}

\end{algorithm}
\caption{Solving \DCC by reduction to \HTNC.}
\label{FIG:PseudocodeReduction-cstn-htn}
\end{figure}
The pseudo-code for checking \eDCC is given in Algorithm~\ref{algo:solve_epsilonzeroDCC}, 
whereas the pseudo-code for checking \DCC is provided in Algorithm~\ref{algo:solve_DCC}. 
The latter algorithm goes as follows. 
Firstly, it computes a sufficiently small $\epsilon>0$ by resorting to Theorem~\ref{thm:epsilonconsistency}, 
\ie $\hat{\epsilon}=|\Sigma_P|^{-1}|V|^{-1}$ (at line~1 of Algorithm~\ref{algo:solve_DCC}).
Secondly, it constructs $\H_{\hat{\epsilon}}(\Gamma)$ (at line~1 of Algorithm~\ref{algo:solve_epsilonzeroDCC}) 
and then it scales every hyperarc's weight to $\Z$ (at lines~2-3).
Thirdly, $\H_{\hat{\epsilon}}(\Gamma)$ is solved with the \HTNC algorithm underlying Theorem~\ref{Teo:MainAlgorithms} (at line~4), 
\ie an instance of the \HTNC problem is solved by reduction to the decision problem for \MPG{s}. 
If the \HTNC algorithm outputs \texttt{YES}, together with a feasible scheduling $\phi$ of $H_{\hat{\epsilon}}(\Gamma)$, 
then the time values of $\phi$ are scaled back to size \wrt $\hat{\epsilon}$ 
and then $\langle \texttt{YES}, \phi \rangle$ is returned as output (lines~5-8); 
otherwise, the output is simply \texttt{NO} (at line~10).
\begin{Rem}
The same algorithm, with essentially the same upper bound on its running time and space,
work also in case we allow for arbitrary boolean formulae as labels, rather than just conjunctions.
At the same time, hyperarc constraints can also be allowed inside the input \CSTN{s}, besides the standard arc constraints. 
Under this prospect, our algorithm actually solves a larger family of conditional temporal networks, 
that we may call \emph{Conditional~Hyper~Temporal~Networks}~(CHyTNs).
\end{Rem}
\begin{Rem}
We remark that the \HTN/\MPG algorithm that is at the heart of our approach requires integral weights (\ie it requires that $w(u,v)\in\Z$ for every $(u,v)\in A$), 
and we could not play it differently~\cite{CPR2014,CPR2015}. Moreover, the algorithm always computes integral solution to \HTN/\MPG{s} and, 
therefore, it always computes rational feasible schedules for 
the \CSTN{s} given in input. 
As such, this ``requirement`` actually turns out to be a plus in practice. 
To conclude, it is indeed integrality that allows us to analyze the algorithm quantitatively 
and to present a sharp lower bounding analysis on the critical value of the reaction time $\hat{\epsilon}$, 
where the \CSTN transits from being, to not being, dynamically consistent. 
We believe that these issues deserve much attention,  
going into them required an algorithmic discrete approach to the notion of numbers.  
\end{Rem}
Now, the correctness and the time complexity of Algorithm~\ref{algo:solve_DCC} is analyzed.
To begin, notice that some of the temporal constraints introduced 
during the reduction step depends on a sufficiently small parameter $\epsilon>0$, 
whose magnitude turns out to depend on the size of the input \CSTN.
It is now proved that the time complexity of the algorithm depends multiplicatively on $D$,  
provided that $\epsilon=N/D$ for some $N,D\in\N_0$.
In Section~\ref{sect:epsilon} we will present a sharp lower bounding analysis on $\hat{\epsilon}$, 
from which the (pseudo) singly-exponential time bound follows as corollary.
So, assume for a moment line~1 to be valid, we prove it in Theorem~\ref{thm:boundexp}.
As a corollary of Theorem~\ref{thm:mainreduction}, we have that Algorithm~\ref{algo:solve_DCC} correctly decides \DCC. 
The most time expensive step of the algorithm is clearly line~4 of Algorithm~\ref{algo:solve_epsilonzeroDCC}, 
which resorts to Theorem~\ref{Teo:MainAlgorithms} in order to solve an instance of \HTNC. 
From Theorem~\ref{thm:mainreduction} we have an upper bound on the size of $\H_{\epsilon}(\Gamma)$, while Theorem~\ref{Teo:MainAlgorithms} gives 
us a pseudo-polynomial upper bound for the computation time. 
Also, recall that we scale weights by a factor $D$ at lines~2-3 of Algorithm~\ref{algo:solve_epsilonzeroDCC}, 
where $\epsilon=N/D$ for some $N,D\in\N_0$. Thus, by composition, Algorithm~\ref{algo:solve_DCC} 
decides \DCC in a time $T_{|\Gamma|}$ which is bounded as follows, where $W\triangleq \max_{a\in A} |w_a|$: 
\[
T_{|\Gamma|}= O((| V_{H_{\epsilon}(\Gamma)} | + |\A_{\H_{\epsilon}(\Gamma)}|) m_{\A_{\H_{\epsilon}}(\Gamma)})WD. 
\]
Whence, the following holds: 
\[T_{|\Gamma|}=O(|\Sigma_P|^{2}|A|^2 + |\Sigma_P|^3|A||V||P| + |\Sigma_P|^4|V|^2|P|)WD.\]

By Theorem~\ref{thm:boundexp}, 
it is sufficient to check $\epsilon$-dynamic-consistency for $\epsilon=|\Sigma_P|^{-1} |V|^{-1}$.
An $O(|\Sigma_P|^{3}|V||A|^2+|\Sigma_P|^4|A||V|^2|P| + |\Sigma_P|^5|V|^3|P|)W$ 
worst-case time bound follows for Algorithm~\ref{algo:solve_DCC}. 
Since $|\Sigma_P|\leq 2^{\min(|P|, l)}$ (where $\ell$ is the number of distinct labels that appear in $\Gamma$), 
the singly-exponential time bound follows. This proves Theorem~\ref{thm:mainresult}.

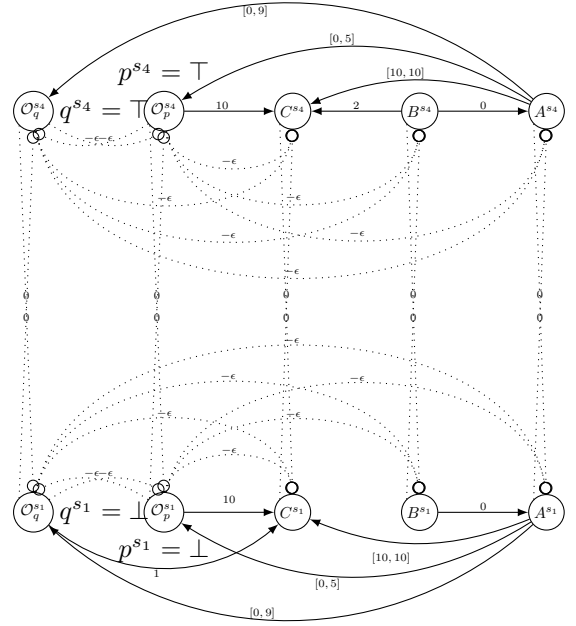
\begin{figure}[!htb]
\centering
\begin{tikzpicture}[arrows=->,scale=.6,node distance=8 and 2]
 	\node[node] (A') {$A^{s_4}$};
	\node[node,xshift=0ex,left=of A'] (B') {$B^{s_4}$};
	\node[node,xshift=0ex,left=of B'] (C') {$C^{s_4}$};
	\node[node,xshift=0ex,left=of C', label={above:$p^{s_4}=\top$}] (P') {$\Ord_p^{s_4}$};
	\node[node,xshift=0ex,left=of P', label={right:$q^{s_4}=\top$}] (Q') {$\Ord_q^{s_4}$};
	
	\node[node, below=of Q',yshift=0ex, label={right:$q^{s_1}=\bot$}] (Q'''') {$\Ord_q^{s_1}$};	
	\node[node, right=of Q'''',yshift=0ex, label={below:$p^{s_1}=\bot$}] (P'''') {$\Ord_p^{s_1}$};
	\node[node, right=of P'''',yshift=0ex] (C'''') {$C^{s_1}$};
	\node[node, right=of C'''',yshift=0ex] (B'''') {$B^{s_1}$};
	\node[node, right=of B'''',yshift=0ex] (A'''') {$A^{s_1}$};

	\draw[] (A') to [bend right=22] node[xshift=-2ex, yshift=0ex,timeLabel,above] {$[10,10]$} (C'); 
	\draw[] (B') to [] node[timeLabel,above] {$2$} (C'); 
	\draw[] (B') to [] node[timeLabel,above] {$0$} (A'); 
	\draw[] (A') to [bend right=32] node[xshift=-2ex, yshift=0ex, timeLabel,above] {$[0,5]$} (P'); 
	\draw[] (A') to [bend right=42] node[xshift=-5ex,yshift=0ex, timeLabel,below] {$[0,9]$} (Q'); 
	\draw[] (P') to [] node[xshift=-1ex, timeLabel,above] {$10$} (C'); 
	
	\draw[] (A'''') to [bend left=22] node[xshift=-4ex, yshift=-1ex,timeLabel,below] {$[10,10]$} (C''''); 
	\draw[] (B'''') to [] node[timeLabel,above] {$0$} (A''''); 
	\draw[] (A'''') to [bend left=32] node[xshift=-4ex, yshift=0ex, timeLabel,below] {$[0,5]$} (P''''); 
	\draw[] (A'''') to [bend left=42] node[xshift=-4ex,yshift=2ex, timeLabel,below] {$[0,9]$} (Q''''); 
	\draw[] (P'''') to [] node[xshift=0ex,yshift=1ex, timeLabel,above] {$10$} (C''''); 
	\draw[] (Q'''') to [bend right=40] node[xshift=-1ex, timeLabel,below] {$1$} (C''''); 

	\draw[>=o, dotted] (A''''.north west) to [] node[xshift=0ex, yshift=2ex,timeLabel,above] {$0$} (A'.south); 
	\draw[>=o, dotted] (P') to [bend right=80] node[xshift=0ex, yshift=0ex,timeLabel,above] {$-\epsilon$} (A'.south); 
	\draw[>=o, dotted] (Q') to [bend right=80] node[xshift=0ex, yshift=0ex,timeLabel,above] {$-\epsilon$} (A'.south); 
	\draw[>=o, dotted] (B''''.north west) to [] node[xshift=0ex, yshift=2ex,timeLabel,above] {$0$} (B'.south); 
	\draw[>=o, dotted] (P') to [bend right=80] node[xshift=0ex, yshift=0ex,timeLabel,above] {$-\epsilon$} (B'.south); 
	\draw[>=o, dotted] (Q') to [bend right=80] node[xshift=0ex, yshift=0ex,timeLabel,above] {$-\epsilon$} (B'.south); 
	\draw[>=o, dotted] (C''''.north west) to [] node[xshift=0ex, yshift=2ex,timeLabel,above] {$0$} (C'.south); 
	\draw[>=o, dotted] (P') to [bend right=80] node[xshift=0ex, yshift=0ex,timeLabel,above] {$-\epsilon$} (C'.south); 
	\draw[>=o, dotted] (Q') to [bend right=80] node[xshift=0ex, yshift=0ex,timeLabel,above] {$-\epsilon$} (C'.south); 
	\draw[>=o, dotted] (P''''.north west) to [] node[xshift=0ex, yshift=2ex,timeLabel,above] {$0$} (P'.south); 
	\draw[>=o, dotted] (Q') to [bend right=30] node[xshift=0ex, yshift=0ex,timeLabel,above] {$-\epsilon$} (P'.south); 
	\draw[>=o, dotted] (Q''''.north west) to [] node[xshift=0ex, yshift=2ex,timeLabel,above] {$0$} (Q'.south); 
	\draw[>=o, dotted] (P') to [bend left=30] node[xshift=0ex, yshift=0ex,timeLabel,above] {$-\epsilon$} (Q'.south); 
	
	\draw[>=o, dotted] (A'.south west) to [] node[xshift=0ex, yshift=-2ex,timeLabel,above] {$0$} (A''''.north); 
	\draw[>=o, dotted] (P'''') to [bend left=80] node[xshift=0ex, yshift=0ex,timeLabel,above] {$-\epsilon$} (A''''.north); 
	\draw[>=o, dotted] (Q'''') to [bend left=80] node[xshift=0ex, yshift=0ex,timeLabel,above] {$-\epsilon$} (A''''.north); 
	\draw[>=o, dotted] (B'.south west) to [] node[xshift=0ex, yshift=-2ex,timeLabel,above] {$0$} (B''''.north); 
	\draw[>=o, dotted] (P'''') to [bend left=80] node[xshift=0ex, yshift=0ex,timeLabel,above] {$-\epsilon$} (B''''.north); 
	\draw[>=o, dotted] (Q'''') to [bend left=80] node[xshift=0ex, yshift=0ex,timeLabel,above] {$-\epsilon$} (B''''.north); 
	\draw[>=o, dotted] (C'.south west) to [] node[xshift=0ex, yshift=-2ex,timeLabel,above] {$0$} (C''''.north); 
	\draw[>=o, dotted] (P'''') to [bend left=80] node[xshift=0ex, yshift=0ex,timeLabel,above] {$-\epsilon$} (C''''.north); 
	\draw[>=o, dotted] (Q'''') to [bend left=80] node[xshift=0ex, yshift=0ex,timeLabel,above] {$-\epsilon$} (C''''.north); 
	\draw[>=o, dotted] (P'.south west) to [] node[xshift=0ex, yshift=-2ex,timeLabel,above] {$0$} (P''''.north); 
	\draw[>=o, dotted] (Q'''') to [bend left=30] node[xshift=0ex, yshift=0ex,timeLabel,above] {$-\epsilon$} (P''''.north); 
	\draw[>=o, dotted] (Q'.south west) to [] node[xshift=0ex, yshift=-2ex,timeLabel,above] {$0$} (Q''''.north); 
	\draw[>=o, dotted] (P'''') to [bend right=30] node[xshift=0ex, yshift=0ex,timeLabel,above] {$-\epsilon$} (Q''''.north); 
\end{tikzpicture}
\caption{An excerpt of the \HTN $\H_{\epsilon}(\Gamma)$ corresponding to the \CSTN $\Gamma$ of Fig.~\ref{FIG:cstn2}, 
in which two scenarios, $s_4$ and $s_1$, are considered and the corresponding hyper-constraints $H(s_4;s_1;u)$ are depicted as dotted hyperarcs.}
\label{FIG:cstn22htn}
\end{figure}

\section{Bounding Analysis on the Reaction Time $\hat{\epsilon}$}\label{sect:epsilon}
In this section we present an asymptotically sharp lower bound for $\hat{\epsilon}(\Gamma)$,  
that is the critical value of reaction time where the \CSTN transits from being, 
to not being, dynamically-consistent. The proof technique introduced in this analysis is applicable more in general, 
when dealing with linear difference constraints which include strict inequalities. 
Moreover, this bound implies that Algorithm~\ref{algo:solve_DCC} is a (pseudo) singly-exponential time algorithm for solving \DCC. 
To begin, we are going to provide a proof of Theorem~\ref{thm:boundexp}, but let us first introduce some further notation.

Let $\Gamma\triangleq\langle V, A, L, \Ord, {\Ord}V, P \rangle$ be a dynamically-consistent \CSTN.
By Theorem~\ref{thm:mainreduction}, there exists $\epsilon>0$ such that $\H_{\epsilon}(\Gamma)$ is consistent.
Then, let $\phi:V^{\text{Ex}}_{\Gamma}\rightarrow\RR$ be a feasible scheduling for $\H_{\epsilon}(\Gamma)$.
For any hyperarc $A=\langle t_A, H_A, w_A\rangle\in\A_{\H_{\epsilon}}$, define a standard arc $a_A$ as follows: 
\[a_A\triangleq \langle t_A, \hat{h}, w_A(\hat{h})\rangle, 
\text{ where } \hat{h}\triangleq \arg\min_{h\in H_A}\big( \phi(h)-w_A(h) \big).\]
Then, notice that the network  
$T^{\phi}_{\epsilon}(\Gamma)\triangleq \langle V^{\text{Ex}}_{\Gamma}, \bigcup_{A\in\A_{\H_\epsilon}} a_A\rangle$ is an \STN.
Moreover, $\phi$ is feasible for $T^{\phi}_{\epsilon}(\Gamma)$.
At this point, assuming $v\in V_\Gamma^{\text{Ex}}$, consider the fractional part $r_v$ of $\phi_v$, 
\ie \[r_v\triangleq \phi_v-\lfloor\phi_v\rfloor.\]
Then, let $R\triangleq \{r_v\}_{v\in V^{\text{Ex}}_\Gamma}$ be the set of all the fractional parts.
Sort $R$ by the common ordering on $\RR$ and  
assume that $S\triangleq \{r_1, \ldots, r_k\}$ is the resulting ordered set without repetitions, 
\ie $|S|=k$, $S=R$, $r_1<\ldots<r_k$.
Now, let $\texttt{pos}(v)$ be the index position such that: \[ 1\leq \texttt{pos}(v)\leq k\text{ and } r_{\texttt{pos}(v)}=r_v.\] 

Then, we define a new fractional part as follows: 
\begin{equation} r'_v\triangleq \frac{\texttt{pos}(v)-1}{|\Sigma_P||V|} \tag{NFP} \end{equation}
also, we define a new scheduling function as follows: \begin{equation} \phi'_v\triangleq \lfloor \phi_v\rfloor + r'_v \tag{NSF} \end{equation}

\begin{Rem}\label{rem:invariant}
Notice that (NFP) doesn't alter the ordering relation among the fractional parts, 
\ie \[ r'_u<r'_v\iff r_u<r_v, \text{ for any } u,v\in V^{\text{Ex}}_\Gamma,\]
moreover, observe that (NSF) doesn't change the value of any integral part, 
\ie \[\lfloor \phi'_u\rfloor = \lfloor\phi_u\rfloor, \text{ for any } u\in V^{\text{Ex}}_\Gamma.\]
\end{Rem}

We are now in the position to prove Theorem~\ref{thm:boundexp}.
\begin{proof}[Proof of Theorem~\ref{thm:boundexp}] 
Let $\Gamma$ be dynamically-consistent, by Theorem~\ref{thm:mainreduction}
there exists $\epsilon'>0$ such that $\H_{\epsilon'}(\Gamma)$ is consistent 
and admits some feasible scheduling $\phi:V^{\text{Ex}}_\Gamma\rightarrow\RR$.
Let $\epsilon\triangleq |\Sigma_P|^{-1}|V|^{-1}$.
We argue that $\phi'$, as defined in (NSF), is a feasible scheduling for the \STN $T^{\phi}_{\epsilon}(\Gamma)$.
Indeed, every difference constraint of $T^{\phi}_{\epsilon}(\Gamma)$ 
is of the form $\phi_v-\phi_u\leq w$, for some $w\in\Z$ or $w=-\epsilon$.
Consider the case $w\in\Z$. Then, $\phi'_v-\phi'_u \leq w$ holds because 
of Remark~\ref{rem:invariant}.
Now, consider the case $w=-\epsilon$. Then, $\phi_v-\phi_u\leq -\epsilon$ implies $\phi_v\neq \phi_u$. 
Hence, by Remark~\ref{rem:invariant}, we have $\phi'_v\neq \phi'_u$. 
At this point, observe that the difference between $\phi'_u$ and $\phi'_v$ 
is therefore at least $\epsilon$, \ie \[\phi'_u-\phi'_v\geq |\Sigma_P|^{-1}|V|^{-1}=\epsilon.\] 
That is to say, $\phi'_v - \phi'_u \leq -\epsilon$. 
This proves that $\phi'$ is a feasible scheduling for the \STN $T^{\phi}_{\epsilon}(\Gamma)$.
Since $T^{\phi}_{\epsilon}(\Gamma)$ is thus consistent, then $\H_{\epsilon}(\Gamma)$ is consistent as well.
Therefore, by Theorem~\ref{thm:mainreduction}, 
the \CSTN $\Gamma$ is $\epsilon$-dynamically-consistent.
\end{proof}

At this point, 
a natural question is whether the lower bound given by Theorem~\ref{thm:boundexp} 
can be improved up to $\hat{\epsilon}(\Gamma)=\Omega(|V|^{-1})$. 
In turn, this would improve the time complexity for Algorithm~\ref{algo:solve_DCC} by a factor $|\Sigma_P|$.
However, the following theorem shows that this is not the case by exhibiting 
a \CSTN for which $\hat{\epsilon}(\Gamma) = 2^{-\Omega(|P|)}$. 
This proves that the lower bound given by Theorem~\ref{thm:boundexp} is (almost) asymptotically sharp. 

\begin{Thm}\label{thm:exponential_epsilon}
For each $n\in N_0$ there exists a \CSTN $\Gamma^n$ such that $\hat{\epsilon}(\Gamma^n) < 2^{-n+1} = 2^{-|P^n|/3+1}$, 
where $P^n$ is the set of boolean variables of $\Gamma^n$.
\end{Thm}

\begin{proof}

For each $n\in\N_0$, define a \CSTN $\Gamma^n\triangleq\langle V^n, A^n, L^n, \Ord^n, {\Ord}V^n, P^n \rangle$ as follows.
See Fig.~\ref{FIG:exponentialepsilon-counterexample-CSTN} for a clarifying illustration.

\begin{figure}[!htb]
\centering
\begin{tikzpicture}[arrows=->,scale=0.88, node distance=1.5 and 1.5]
	\node[node, label={above:$Y_1?$}] (Y1) {$Y_1$};
 	\node[node, above left = of Y1, xshift=0ex, label={left:$\textbf{0}$}, label={above:\small $X_1?$}] (X1) {$X_1$};
	\node[node, above right = of Y1, xshift=0ex, label={above:\small $Z_1?$}] (Z1) {$Z_1$};
	\node[node, below = of Y1, yshift=-10ex, label={above:\small $Y_2?$}] (Y2) {$Y_2$};
	\node[node, above left = of Y2, xshift=0ex, label={above:\small $X_2?$}] (X2) {$X_2$};
	\node[node, above right = of Y2, xshift=0ex, label={above:\small $Z_2?$}] (Z2) {$Z_2$};
	\node[below left = of Y2, xshift=0ex] (fakeUL) {};
	\node[below right = of Y2, xshift=0ex] (fakeUR) {};
	\node[node, blackNode, scale=0.2, below = of Y2] (p1) {};
	\node[node, blackNode, scale=0.2, below = of p1, yshift=20ex] (p2) {};
	\node[node, blackNode, scale=0.2, below = of p2, yshift=20ex] (p3) {};
	\node[below left = of p3, yshift=8ex] (fakeDL) {};
	\node[below = of p3, yshift=8ex] (fakeDC) {};
	\node[below right = of p3, yshift=8ex] (fakeDR) {};
	\node[node, below = of fakeDC, yshift=-15ex, label={above:\small $Y_n?$}] (Yn) {$Y_n$};
	\node[node, above left = of Yn, xshift=0ex, label={above:\small $X_n?$}] (Xn) {$X_n$};
	\node[node, above right = of Yn, xshift=0ex, label={above:\small $Z_n?$}] (Zn) {$Z_n$};
	\draw[] (X1) to [] node[xshift=0ex, yshift=0ex,above] {\footnotesize $1, X_1 Y_1$} (Z1); 
	\draw[] (X1) to [] node[xshift=-3ex, yshift=0ex,below] {\footnotesize $[2], \neg X_1$} (Y1); 
	\draw[] (Y1) to [] node[xshift=3ex, yshift=0ex,below] {\footnotesize $[2], \neg Y_1$} (Z1); 
	\draw[] (X1) to [bend right=50] node[xshift=-4ex, yshift=0ex,above] {\footnotesize $[5], Z_1$} (X2);
	\draw[] (Z1) to [bend left=50] node[xshift=6ex, yshift=0ex,above] {\footnotesize $[5], \neg Z_1 X_2 Y_2$} (Z2);  
	\draw[] (Y1) to [] node[xshift=-3ex, yshift=0ex,above] {\footnotesize $[5], \neg Z_1$} (X2);  
	\draw[] (Y1) to [] node[xshift=4ex, yshift=0ex,above]  {\footnotesize $[5], Z_1 X_2 Y_2$} (Z2);  
	\draw[] (X2) to [] node[xshift=-3ex, yshift=0ex,below] {\footnotesize $[2], \neg X_2$} (Y2);  
	\draw[] (Y2) to [] node[xshift=3ex, yshift=0ex,below] {\footnotesize $[2], \neg Y_2$} (Z2);  
	\draw[] (Y2) to [] node[xshift=-3ex, yshift=0ex,above] {\footnotesize $[5], \neg Z_2$} (fakeUL); 
	\draw[] (Y2) to [] node[xshift=4ex, yshift=0ex,above] {\footnotesize $[5], Z_2 X_3 Y_3$} (fakeUR); 
	\draw[] (X2) to [bend right=50] node[xshift=-4ex, yshift=0ex,above] {\footnotesize $[5], Z_2$} (fakeUL); 
	\draw[] (Z2) to [bend left=50] node[xshift=6ex, yshift=0ex,above] {\footnotesize $[5], \neg Z_2 X_3 Y_3$} (fakeUR); 
	\draw[] (Xn) to [] node[xshift=-3ex, yshift=0ex,below] {\footnotesize $[2], \neg X_n$} (Yn); 
	\draw[] (Yn) to [] node[xshift=3ex, yshift=0ex,below] {\footnotesize $[2], \neg Y_n$} (Zn); 
	\draw[] (fakeDL) to [bend right=50] node[xshift=-4ex, yshift=1ex,above] {\footnotesize $[5], Z_{n-1}$} (Xn); 
	\draw[] (fakeDR) to [bend left=50] node[xshift=7ex, yshift=1ex,above] {\footnotesize $[5], \neg Z_{n-1} X_n Y_n$} (Zn); 
	\draw[] (fakeDC) to [] node[xshift=-2ex, yshift=-1ex,above] {\footnotesize $[5], \neg Z_{n-1}$} (Xn); 
	\draw[] (fakeDC) to [] node[xshift=2ex, yshift=-1ex,above] {\footnotesize $[5], Z_{n-1} X_n Y_n$} (Zn); 
\end{tikzpicture}
\caption{A CSTN $\Gamma^n$ such that $\hat{\epsilon}(\Gamma^n) = 2^{-\Omega(|P^n|)}$.}
\label{FIG:exponentialepsilon-counterexample-CSTN}
\end{figure}
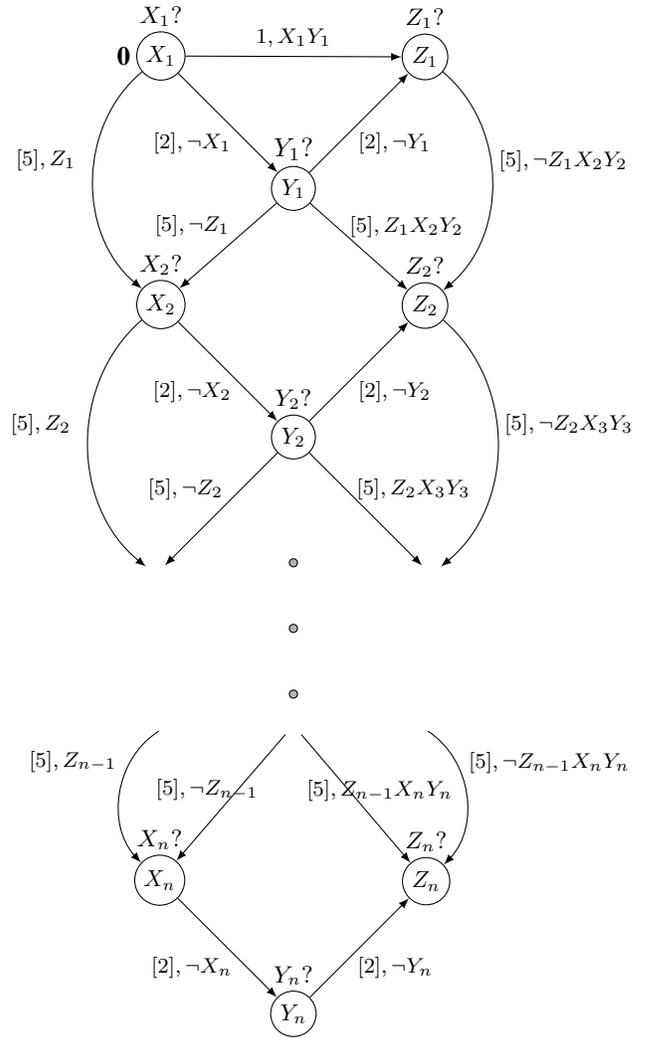

\begin{itemize}
\item $V^n\triangleq\{X_i, Y_i, Z_i \mid 1\leq i\leq n\}$;
\item $A^n\triangleq B \cup \bigcup_{i=1}^{n} C_i\cup \bigcup_{i=1}^{n-1} D_i \cup E$ \\ where:
\begin{itemize}
\item 
$B\triangleq \{\langle X_1-v\leq 0, \lambda\rangle \mid v\in V^n\}\cup \{\langle Z_1 - X_1 \leq 1, X_1\wedge Y_1\rangle\}$;
\item 
$C_i \triangleq \{\langle Y_i - X_i \leq 2, \neg X_i\rangle, \langle X_i - Y_i \leq -2, \neg X_i\rangle,  
			\langle Z_i - Y_i \leq 2, \neg Y_i\rangle, \langle Y_i - Z_i \leq -2, \neg Y_i\rangle\}$;
\item 
$D_i \triangleq \{\langle X_{i+1} - X_i \leq 5, Z_i\rangle, \langle X_i - X_{i+1}\leq -5, Z_i\rangle, 
                         \langle X_{i+1} - Y_i\rangle\leq 5, \neg Z_i\rangle, \langle Y_i - X_{i+1}\leq -5, \neg Z_i\rangle,  
		        \langle Z_{i+1}-Y_i\leq 5, Z_i\wedge X_{i+1}\wedge Y_{i+1}\rangle, \langle Y_i-Z_{i+1}\leq -5, Z_i\wedge X_{i+1}\wedge Y_{i+1}\rangle, 
			 \langle Z_{i+1} - Z_i \leq 5, \neg Z_i \wedge X_{i+1}\wedge Y_{i+1}\rangle, 
                         \langle Z_i - Z_{i+1}\leq -5, \neg Z_i \wedge X_{i+1}\wedge Y_{i+1}\rangle \}$;
\item 
$E \triangleq \{\langle Y_n - X_n \leq 2, \neg X_n\rangle, \langle X_n - Y_n \leq -2, \neg X_n\rangle, 	
			\langle Z_n-Y_n \leq 2, \neg Y_n\rangle, \langle Y_n - Z_n \leq -2, \neg Y_n\rangle \}$;
\end{itemize}
\item $L^n(v)\triangleq\lambda$ for every $v\in V^n$; 
${\Ord}V^n\triangleq V^n$; 
$\Ord^n(v)\triangleq v$ for every $v\in {\Ord}V^n$;
$P^n\triangleq V^n$.
\end{itemize}

We exhibit a viable and dynamic execution strategy $\sigma_n:\Sigma_{P^n}\rightarrow \Phi_{V^n}$ for $\Gamma^n$.

Let $\{\delta_i\}_{i=1}^{n}$ and $\{\Delta_i\}_{i=1}^{n}$ be two real valued sequences s.t.:
\[(1)\;\Delta_1\triangleq 1; (2)\; 0<\delta_i<\Delta_i; 
(3)\; \Delta_{i}\triangleq \min(\delta_{i-1}, \Delta_{i-1}-\delta_{i-1}).\]
Then, the following also holds for every $1\leq i\leq n$: 
\[ (4)\; 0 < \Delta_i \leq 2^{-i+1}, \]
where the equality holds if and only if $\delta_i = \Delta_i/2$.

In what follows, provided that $s\in\Sigma_P$ and $\ell\in P^*$, 
we will denote $\mathds{1}_{s(\ell)}\triangleq 1$ if $s(\ell)=\top$ and $\mathds{1}_{s(\ell)}\triangleq 0$ if $s(\ell)=\bot$.

We are ready to define $\sigma_n(s)$ for any $s\in\Sigma_P$: 
\begin{itemize}
\item $[\sigma_n(s)]_{X_1}\triangleq 0$;
\item $[\sigma_n(s)]_{Y_1}\triangleq \delta_1 \mathds{1}_{s(X_1)} + 2\mathds{1}_{s(\neg X_1)}$; 
\item $[\sigma_n(s)]_{Z_1}\triangleq$
	$\mathds{1}_{s(X_1\wedge Y_1)} + (2+[\sigma_n(s)]_{Y_1})\mathds{1}_{s(\neg X_1 \vee \neg Y_1)}$;
\item $[\sigma_n(s)]_{X_{i}}\triangleq 5 + [\sigma_n(s)]_{X_{i-1}}\mathds{1}_{s(Z_{i-1})} + \\ 
	+ [\sigma_n(s)]_{Y_{i-1}}\mathds{1}_{s(\neg Z_{i-1})}$, for any $2\leq i\leq n$;
\item $[\sigma_n(s)]_{Y_i}\triangleq$ $[\sigma_n(s)]_{X_i} + \delta_i\mathds{1}_{s(X_i)} + 2\mathds{1}_{s(\neg X_i)}$, for any $2\leq i \leq n$;
\item $[\sigma_n(s)]_{Z_{i}}\triangleq \big(5+[\sigma_n(s)]_{Y_{i-1}}\mathds{1}_{s(Z_{i-1})} + \\
                + [\sigma_n(s)]_{Z_{i-1}}\mathds{1}_{s(\neg Z_{i-1})}\big)\mathds{1}_{s(X_{i}\wedge Y_{i})} + \\ 
                + (2+[\sigma_n(s)]_{Y_{i}})\mathds{1}_{s(\neg X_i\vee \neg Y_i)}$, for any $2\leq i\leq n$;
\end{itemize}
It is not difficult to prove, by induction on $n\geq 1$, that $\sigma_n$ is viable and dynamic for $\Gamma^n$.

Here we show that $\hat{\epsilon}(\Gamma^n) < 2^{-n+1}=2^{-|P^n|/3+1}$ for every $n\geq 1$.
Let us consider the following scenario $\hat{s}$ for $1\leq i\leq n$: 
\[
\hat{s}(X_i)\triangleq \hat{s}(Y_i)\triangleq \top; \;\;\; 
\hat{s}(Z_i)\triangleq \left\{ 
\begin{array}{ll}
\top, & \text{ if } \delta_i \leq \Delta_i/2 \\
\bot, & \text{ if } \delta_i > \Delta_i/2 \\
\end{array} \right. . 
\]
We assume that $\sigma$ is an execution strategy for $\Gamma^n$ 
and study necessary conditions to ensure that $\sigma$ is viable and dynamic, 
provided that the observations follow scenario $\hat{s}$.
First, $\sigma$ must schedule $X_1$ at time $[\sigma(\hat{s})]_{X_1}=0$. 
Then, since $\hat{s}(X_1)=\top$, we must have $0<[\sigma(\hat{s})]_{Y_1}<1$, because of the constraint $(Z_1-X_1\leq 1, X_1\wedge Y_1)$.
Stated otherwise, it is necessary that: 
\[0 < [\sigma(\hat{s})]_{Y_1} - [\sigma(\hat{s})]_{X_1} < \Delta_1.\]
After that, since $\hat{s}(Y_1)=\top$, then $\sigma$ must schedule $Z_1$ at time $[\sigma(\hat{s})]_{Z_1}=1=\Delta_1$.
A moment's reflection reveals that almost identical necessary conditions now recur for $X_2, Y_2, Z_2$, 
with the crucial variation that it will be necessary to require: $0 < [\sigma(\hat{s})]_{Y_2} < \Delta_2$. 
Indeed, proceeding inductively, it will be necessary that for every $1\leq i\leq n$ and every $n\in\N_0$: 
\[0 < [\sigma(\hat{s})]_{Y_i} - [\sigma(\hat{s})]_{X_i} < \Delta_i.\] 
As already observed in ($4$), we have $0<\Delta_n\leq 2^{-n+1}$. 
Thus, any viable and dynamic execution strategy $\sigma$ for $\Gamma^n$ must satisfy: 
\[\displaystyle0 < [\sigma(\hat{s})]_{Y_n} - [\sigma(\hat{s})]_{X_n} < \frac{1}{2^{n-1}}=\frac{1}{2^{|P^n|/3-1}}.\]
Thus, once the planner has observed the outcome $\hat{s}(X_n)=\top$ from the observation event $X_n$, 
then he must react by scheduling $Y_n$ within time $2^{-n+1}=2^{-|P^n|/3+1}$ in the future \wrt $[\sigma(\hat{s})]_{X_n}$.
Then $\hat{\epsilon}(\Gamma^n) < 2^{-n+1}=2^{-|P^n|/3+1}$ any $n\geq 1$.
\end{proof}

\section{Related Works}\label{sect:relatedworks} 
This section discusses of some alternative approaches offered by the current literature.
Recall that the article of Tsamardinos,~\etal~\cite{TVP2003} has been discussed already in the introduction.
The work of Cimatti,~\etal~\cite{Ci14} provided the first sound-and-complete algorithm for checking
the dynamic-controllability of \CSTN{s} with Uncertainty (CSTNU) 
and thus it can be employed for checking the dynamic-consistency of \CSTN{s} as a special case. 
The algorithm reduces to the problem of solving Timed Game Automata (TGA). 
Nevertheless, no worst-case bound on the time complexity of the procedure was provided in~\cite{Ci14}.
We observe that solving TGAs is a problem of much higher complexity than solving \MPG{s}, 
compare the following known facts:
solving 1-player TGAs is $\PSPACE$-complete and solving 2-player TGAs is $\EXP$-complete; 
on the contrary, the problem of determining \MPG{s} lie in $\NP\cap\coNP$ 
and it is currently an open problem to prove whether it lies in $\P$.
Indeed, the algorithm in~\cite{Ci14} is not singly-exponential time bounded.
Finally, a sound algorithm for checking the dynamic-controllability 
of CSTNUs was given by Combi, Hunsberger, Posenato~in~\cite{CHP13}.
However, it was not shown to be complete. To the best of our knowledge, 
it is currently open whether or not it can be extended in order to prove completeness.

\section{Conclusion}\label{sect:conclusions}
We gave the first singly-exponential time algorithm to check the 
dynamic-consistency of \CSTN{s}, also yielding dynamic execution strategies. 
The algorithm actually manages a few more general variants of the problem, 
where labels are not required to be conjunctions and hyperarc constraints can be empolyed in the input \CSTN{s}, 
besides the classical binary constraints. To summarize, at the heart of the algorithm a reduction to \MPG{s} is 
mediated by the \HTN model. The \CSTN is dynamically-consistent if and only if the corresponding \MPG is 
everywhere won, and a dynamic execution strategy can be conveniently read out by an everywhere winning positional strategy.
The size of this \MPG is at most 
polynomial in the number of the possible scenarios; as such, the term at the exponent is linear, 
at worst, in the number of the observation events. The same holds for the running time of the resulting algorithm.
In future works we would like to settle the exact computational complexity of \DCC, 
as well as to extend our approach in order to check the dynamic-controllability of \CSTN 
with Uncertainty~\cite{HPC12}. Finally, an extensive experimental evaluation is on the way.
\paragraph*{Acknowledgment}
This work was partially supported by \emph{Department of Computer Science, University of Verona, Italy} 
under Ph.D. grant ``Computational Mathematics and Biology``.

\bibliographystyle{plain}
\bibliography{biblio}
\end{document}